%% file: neurips_2026.tex
\documentclass{article}

\usepackage[main, final]{neurips_2026}

\usepackage[utf8]{inputenc} % allow utf-8 input
\usepackage[T1]{fontenc}    % use 8-bit T1 fonts
\usepackage{hyperref}       % hyperlinks
\usepackage{url}            % simple URL typesetting
\usepackage{booktabs}       % professional-quality tables
\usepackage{amsfonts}       % blackboard math symbols
\usepackage{nicefrac}       % compact symbols for 1/2, etc.
\usepackage{microtype}      % microtypography
\usepackage{xcolor}         % colors

\input{preamble}

\title{A Theoretical Analysis of Test-Driven Code Generation}

\author{
Nicolas Menet$^{1,2}$\\
{\tt\footnotesize nicolas.menet@ibm.com}\\
\And
Michael Hersche$^{1}$\\
{\tt\footnotesize michael.hersche@ibm.com}\\
\And
Andreas Krause$^{2}$\\
{\tt\footnotesize krausea@ethz.ch}\\
\And
Abbas Rahimi$^{1}$\\
{\tt\footnotesize abr@zurich.ibm.com}
\And
{
\normalfont $^{1}$IBM Research -- Zurich, $^{2}$Department of Computer Science, ETH Zürich}
}

\begin{document}

\maketitle

\input{content}

\newpage
\input{checklist.tex}

\end{document}

%% file: preamble.tex
\usepackage{graphicx, subcaption, wrapfig}
\usepackage{mathtools, dsfont, amsthm}
\usepackage{listings, upquote}
\theoremstyle{plain}
\newtheorem{theorem}{Theorem}[section]
\newtheorem{proposition}[theorem]{Proposition}
\newtheorem{lemma}[theorem]{Lemma}
\newtheorem{corollary}[theorem]{Corollary}
\theoremstyle{definition}
\newtheorem{definition}[theorem]{Definition}

\theoremstyle{remark}

\newtheoremstyle{nonotebrackets}
  {\topsep}   % Space above
  {\topsep}   % Space below
  {\itshape}  % Body font
  {}          % Indent amount
  {\bfseries} % Theorem head font
  {.}         % Punctuation after theorem head
  {.5em}      % Space after theorem head
  {\thmname{#1}\thmnumber{ #2}\thmnote{ #3}} % <--- The magic happens here!
\theoremstyle{nonotebrackets}
\newtheorem*{theorem*}{Theorem}
\newtheorem*{proposition*}{Proposition}

\newenvironment{rebuttal}
{} % \color{blue}
{}

\usepackage{pgfplots}
\pgfplotsset{compat=1.18}

%% file: content.tex
\begin{abstract}
  Code assistants are increasingly utilized in test-driven software development, yet the theoretical mechanisms behind their environment-interaction strategies remain underexplored. We provide a probabilistic framework for two dominant paradigms: code selection after generation using the execution environment, and code generation conditioned on environment feedback. First, we formalize several well-established selection heuristics as environment-aware estimators of code correctness. We theoretically prove that estimators based on fuzzy functional similarity add an inductive bias and strictly dominate estimators based on functional equivalence in terms of signal-to-noise ratio. Second, we frame backprompting as an in-context approximation of Thompson sampling. We derive a novel regret bound for reward functions with unobservable components, theoretically explaining why the effectiveness of backprompting is limited by the ambiguity of the informal task description (an irreducible regret). Using \begin{rebuttal}seven\end{rebuttal} state-of-the-art open weight models, we corroborate these findings across BigCodeBenchHard, LeetCodeDataset, QiskitHumanEvalSim\begin{rebuttal}, and SWE-bench Verified\end{rebuttal}. Our formalization also suggests how to improve task descriptions effectively, leading to an \begin{rebuttal}improved\end{rebuttal} benchmark \begin{rebuttal}for LLM-assisted quantum programming\end{rebuttal}, QiskitHumanEvalSimX.
\end{abstract}

\section{Introduction}
Large Language Models (LLMs) have demonstrated strong capabilities across several domains \citep{achiam2023gpt}, including coding \citep{chen2021evaluating}. Building on these capabilities, code assistants have been integrated into disruptive products, such as GitHub Copilot, Claude Code, and Cursor.

With code assistants increasingly adopted for critical large-scale software, their capability to conduct test-driven development \citep{beck2003test} becomes paramount. In addition to facilitating collaboration and quality control, test-driven development provides a unique opportunity for boosting LLM performance via post-generation selection \citep{li2022competition, shi2022mbr, chen2023codet, mundler2024swt, li2025testtimescal, liu2026dynamic} or iterative backprompting \citep{shinn2023reflexion, olausson2024is, chen2024teaching, li2026oracle}. Environment-aware post-generation selection leverages test suites to score implementations and group equally (or similarly) behaving code. In contrast, backprompting of execution results provides the language model with an opportunity to refine its own beliefs on the behavior of code. These strategies are illustrated in Figure~\ref{fig:overview}. Although widely used, their theoretical mechanisms remain underexplored, resulting in a zoo of heuristics without a common foundation.

In this work, we adopt a probabilistic perspective on test-driven generation. Our contributions are:\looseness -1\\(1) We formalize post-generation selection heuristics as estimators of environment-aware probability of correctness, introducing a unifying categorization of prior works in Table~\ref{tab:types_of_self_consistency}. We prove that estimators adopting functional similarity add an inductive bias that diminishes neural approximation artifacts and strictly dominate estimators built on functional equivalence in terms of signal-to-noise ratio. This provides a theoretical basis for why aggregating ``similar'' behaviors is more robust than counting identical outputs.\\ (2) We explore the limitations of conditioning the LLM on environment feedback during generation. To that end, we frame backprompting as an in-context approximation of Thompson sampling. We then derive a novel regret bound for reward functions with unobservable components. The bound identifies an irreducible regret: the effectiveness of backprompting is strictly limited by the ambiguity of the task description, regardless of the quantity of execution feedback. \begin{rebuttal}Lastly, we propose a practical estimator of the irreducible regret that may be used to trigger clarification requests.\end{rebuttal}
\begin{wrapfigure}{r}{0.535\textwidth}
    \vspace{1pt}
    \centering
    \begin{subfigure}[b]{0.51\linewidth}
        \centering
        \captionsetup{justification=centering}\includegraphics[width=\linewidth]{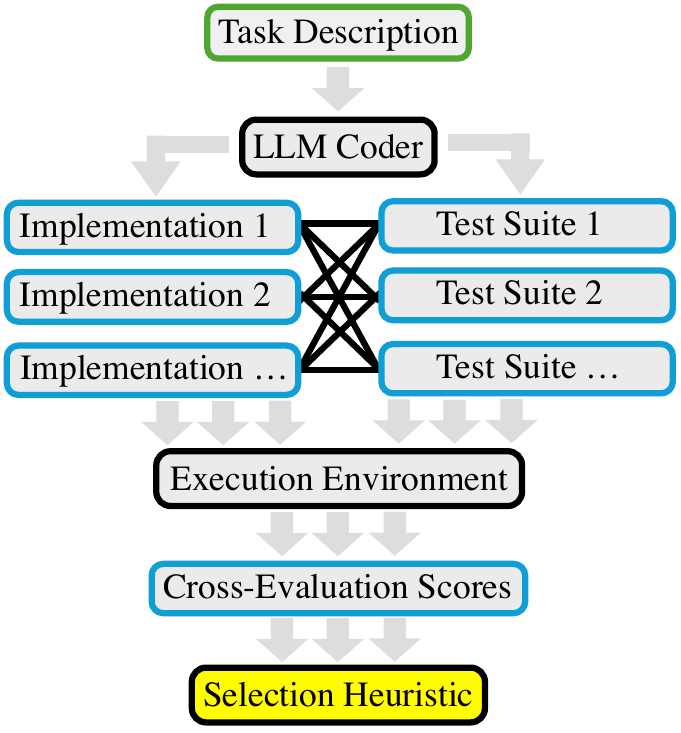}
    \caption{Section~\ref{sec:post_generation_selection_via_feedback}: Selection via Environment after Generation}
    \end{subfigure}
    \hfill
    \begin{subfigure}[b]{0.46\linewidth}
        \centering
        \captionsetup{justification=centering}\includegraphics[width=0.9\linewidth]{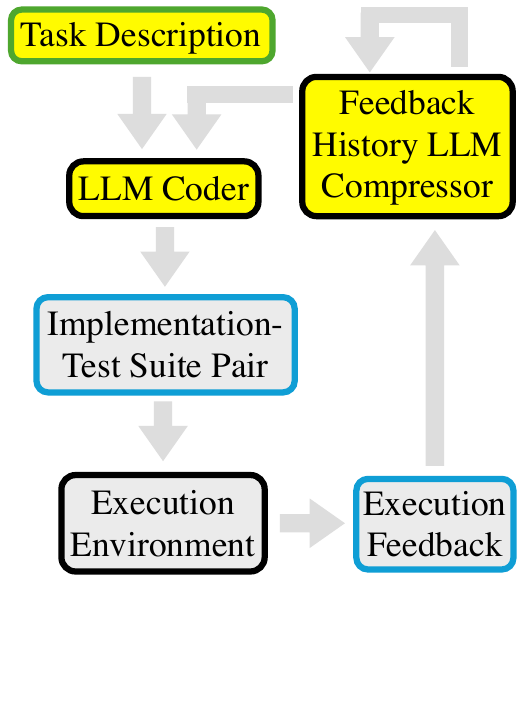}
        \caption{Section~\ref{sec:feedback_during_generation}: Feedback during Generation}
    \end{subfigure}
    \caption{We consider environment interaction in two settings. The scope of our analysis is marked in yellow.}
    \label{fig:overview}
    \vspace{-15pt}
\end{wrapfigure}
(3) We confirm these findings experimentally using \begin{rebuttal}seven\end{rebuttal} open weight models (Qwen3-235B-Instruct, Qwen3-Next, Qwen3-Next-Coder\begin{rebuttal}, Qwen3.5-397B\end{rebuttal}, GPT-OSS-120B, MiniMax-M2.1, \begin{rebuttal} MiniMax-M2.7\end{rebuttal}) on competitive coding (BigCodeBenchHard, LeetCodeDataset, and QiskitHumanEvalSim) \begin{rebuttal}and agentic coding (SWE-bench Verified)\end{rebuttal}. Our results confirm that ``soft'' estimators (adopting functional similarity) consistently outperform ``hard'' ones (adopting functional equivalence), and that backprompting is most effective in settings with a small unobservable reward component. Finally, we verify the importance of clear task descriptions for effective backprompting by deriving a novel dataset \textsc{QiskitHumanEvalSimX}.

\section{Preliminaries}

\paragraph{Notation} 
Let $\mathcal V$ be a finite vocabulary of tokens and $\mathcal{V}^{\le L}$ the set of all token sequences up to a maximum output length $L$. We denote the number of test cases in a test suite $t \in \mathcal V^{\le L}$ by $|t|$.

\begin{wrapfigure}{r}{0.53\textwidth}
    \centering
    \vspace{-16pt}
    \includegraphics[width=\linewidth]{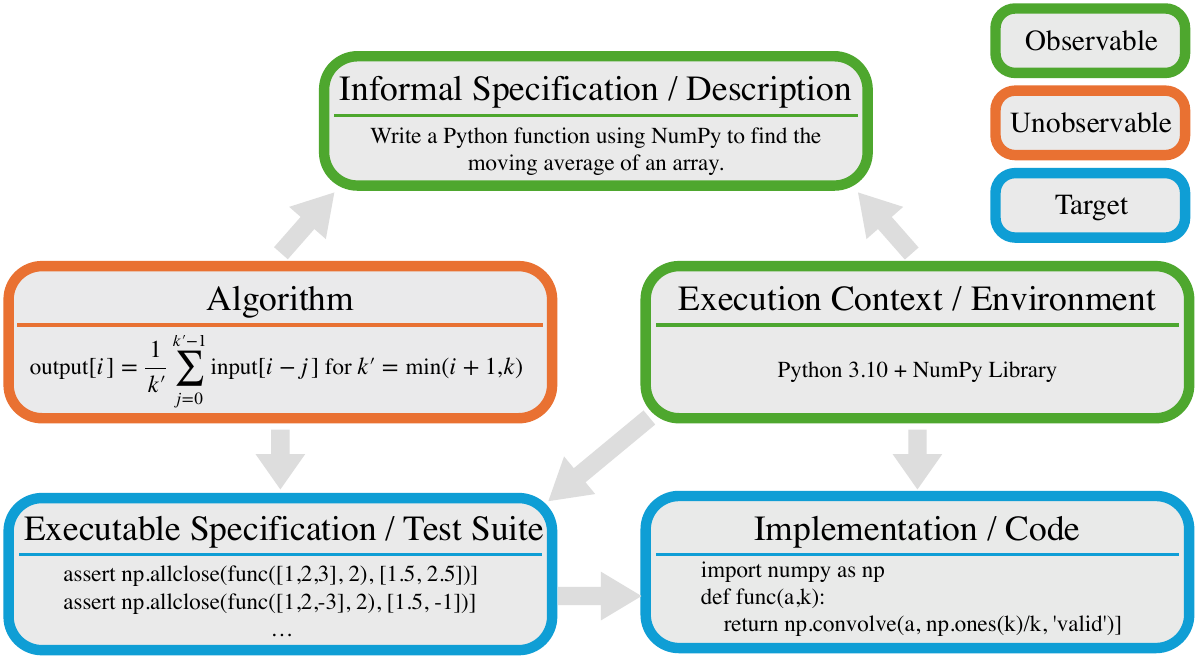}
    \caption{Generative process represented as a structural causal model (arrows indicate causality). We observe an informal specification (description) corresponding to an algorithm in an execution context (environment). We seek an executable specification (test suite) and an implementation that satisfies it (code).}
    \label{fig:probabilistic_model}
\end{wrapfigure}
\paragraph{Problem setting} 
We assume a generative process driven by a prior distribution over abstract \emph{algorithms} $a \in \mathcal A$ and \emph{environments} $e \in \mathcal E$ (execution contexts).
The environment constrains the available tools (e.g., libraries) and required syntax to implement the algorithm. We observe a natural language task \emph{description} $d \in \mathcal V^{\le L}$ (the informal specification), which is a lossy compression of the algorithm $a$ and environment $e$. Crucially, the abstract algorithm $a$ is inherently unobservable beyond $d$. In contrast, the environment $e$ can be partially observed. Given a description $d$ and an observation history $\mathcal H$, we seek two latent variables: a \emph{test suite} $t \in \mathcal V^{\le L}$ (the executable specification) that formalizes the requirements of algorithm $a$ in environment $e$, and an \emph{implementation} $c \in \mathcal V^{\le L}$ valid under this test suite. We model $p(c,t \mid \mathcal H, d)$ with a large language model and shorten the notation by omitting the dependency on $d$. Figure~\ref{fig:probabilistic_model} represents the generative process as a structural causal model.

\paragraph{Environment feedback} Given an environment $e \in \mathcal E$, code $c \in \mathcal V^{\le L}$, and test suite $t \in \mathcal V^{\le L}$, we consider both scalar and textual environment feedback. Let $\mathbb O$ be a finite vocabulary of test outputs. Then, $O(c,t \,\vert\, e) \in \mathbb O^{|t|}$ is the direct output of the test harness. Typically, $\mathbb O = \{0,1\}$ with $0$ indicating failure and $1$ success. In this case, a scalar reward is recovered via $R(c,t \,\vert\, e) := \frac{1}{|t|}\sum_i O_i(c,t \,\vert\, e) \in [0,1]$. More generally, strategies such as functional consensus rely on generalized test suites with large $\mathbb O$, e.g., corresponding to arbitrary bounded function output. Lastly, more informative textual feedback can be obtained from a test report $U(c,t \,\vert\, e) \in \mathcal V^{\le L}$, e.g., in the form of captured \textit{stderr} and assertion tracebacks. In the following, we adopt an environment-centric view to first formalize popular post-generation selection mechanisms that leverage the environment to establish fine-grained consensus, before examining in-the-loop strategies based on backprompting of environment feedback.

\section{Selection via environment after generation}\label{sec:post_generation_selection_via_feedback}
Code generation should target the underlying \textit{algorithm} rather than a specific implementation. Yet, autoregressive next-token prediction leads to probability mass being fragmented across functionally identical implementations---an issue known as \textit{program aliasing} \citep{bunel2018leveraging, zhong2018semregex}. To remedy such modeling artifacts, selection strategies must operate at an algorithmic granularity that groups equally behaving codes. In this section, we first define key concepts such as functional similarity and functional equivalence. We then establish important properties such as their signal-to-noise ratio, and conclude that functional similarity is significantly more robust. Based on our theory, we then unify several prior works in Table~\ref{tab:types_of_self_consistency}. All proofs, including proof sketches, are in Appendix~\ref{sec:proofs}. Moreover, the appropriateness of each adopted assumption is discussed in Section~\ref{sec:discussion_assumptions}.

\begin{definition} [Functional code $s$-similarity]\label{def:functional_code_similarity} Given an environment $e \in \mathcal E$, a distribution over test suites $p(t)$, and a sharpness exponent $s\in\mathbb N$, we define similarity between implementations $c_1, c_2$ as
    \begin{equation*}
        \textstyle \mathrm{sim}_{p,e}^s(c_1, c_2) := \big(\mathbb E_{t\sim p}[ \frac{1}{|t|}\sum_{k=1}^{|t|} \mathds{1}_{O_k(c_1, t \vert e)=O_k(c_2, t \vert e)}]\big)^s,
    \end{equation*}
    where $O_k(c,t\vert e)$ denotes the output of the $k$-th test case in a test suite $t$ for the implementation $c$.
\end{definition}

Note that this similarity implicitly depends on the task description $d$ via $p(t)$. This dependency focuses the functional comparison on the intended input domain, diminishing fragmentation due to divergent handling of edge cases undefined by $d$, and reducing variance of Monte Carlo estimation.

\begin{definition}[Functional code equivalence]\label{def:functional_equivalence} Given an environment $e \in \mathcal E$ and a distribution over test suites $p(t)$, we define the equivalence of two implementations $c_1, c_2$ as  
$$\mathrm{sim}^\infty_{p,e} (c_1,c_2):= \mathds{1}_{O(c_1, t \vert e)=O(c_2, t \vert e)\ \forall t\in \mathcal V^{\le L} : p(t) > 0}.$$
\end{definition}

The notation is deliberate. In addition to a kernel structure that will be central to proving Theorem~\ref{thm:measure_smoothing}, we can prove that functional $s$-similarity converges to functional equivalence as $s \to \infty$.

\begin{proposition}[Properties of $s$-similarity]\label{prop:properties_s_similarity}
    The similarity $\mathrm{sim}_{p,e}^s(c_1, c_2)$ is a positive semi-definite kernel with $\mathrm{sim}_{p,e}^s \geq 0$ and $\mathrm{sim}_{p,e}^s(c,c) = 1$. Moreover, $\mathrm{sim}^\infty_{p,e} (c_1,c_2) = \lim_{s \to \infty} \mathrm{sim}^s_{p,e} (c_1,c_2)$.
\end{proposition}

\subsection{Functional similarity neighborhoods induce measure smoothing}\label{sec:environment_induced_measure_smoothing}

Functional equivalence is the canonical structure we must impose on $\mathcal V^{\le L}$ to factor out redundant representations, but functional similarity offers a smoother objective that can be used to reduce approximation error (see Section~\ref{sec:similarity_vs_equivalence}). Functional code similarity defines similarity neighborhoods:

\begin{definition}[Fuzzy similarity neighborhood]
    For any code $c \in \mathcal V^{\le L}$, let $\mathcal{N}_c^s$ be the fuzzy set representing the neighborhood of implementations functionally similar to implementation $c$. It is defined by the membership function $\mu_{\mathcal{N}_c^s} : \mathcal V^{\le L} \to [0, 1]$ where $\mu_{\mathcal{N}_c^s}(x) := \mathrm{sim}^s_{p,e}(c, x)$.
\end{definition}
For functional equivalence ($s \to \infty$), the sets cease to be fuzzy and become strict equivalence classes.
\begin{proposition}[Functional equivalence classes]\label{prop:equivalence_classes}
  Define $\mathcal N_c^\infty = \{c^\prime \in \mathcal V^{\le L} : \mathrm{sim}^\infty_{p,e}(c, c^\prime) = 1 \}$. Then $\{\mathcal N_c^\infty : c \in \mathcal V^{\le L}\}$ is a set of equivalence classes partitioning $\mathcal V^{\le L}$.  
\end{proposition}

We now lift the concept of functional $s$-similarity to the distribution over implementations $p(c)$.

\begin{definition}[Smoothed code distribution]\label{def:similarity_smoothed_score}
    Given a distribution over implementations $p(c)$, we define the probability measure of a fuzzy similarity neighborhood $\mathcal N_c^s$ as expected membership:
    \begin{equation*}
        \textstyle p(\mathcal N_c^s) := \mathbb E_{c^\prime \sim p}[\mu_{\mathcal{N}_c^s}(c^\prime)] = \sum_{c^\prime \in \mathcal V^{\le L}} p(c^\prime)\mathrm{sim}^s_{p,e}(c, c^\prime).
    \end{equation*}
\end{definition}

This measure aggregates the probability mass of all functionally similar implementations, effectively smoothing the distribution over codes. It allows us to consider the probability of program behavior rather than focusing on specific token strings. In the limit of strict equivalence, $p(\mathcal N_c^\infty)$ recovers the probability of strict behavioral equivalence classes. However, as we discuss next, Monte Carlo estimation of strict equivalence classes is biased and suffers from high variance, whereas smoothing reduces the former by linearization and the latter by aggregating information from similar implementations.

\subsection{Estimating functional similarity is more robust than estimating functional equivalence}\label{sec:similarity_vs_equivalence}
In practice, precisely computing $\mathrm{sim}_{p,e}^s$ and $p(\mathcal N_c^s)$ is prohibitively expensive due to requiring integration over the combinatorial $\mathcal V^{\le L}$. An effective fallback is provided by Monte Carlo estimation:
\begin{definition}
    For $c_i \sim p(c)$ and $t_j \sim p(t)$ i.i.d., we define the Monte Carlo estimators: 
    \begin{align*}
    \textstyle \widehat{\mathrm{sim}}_{p,e;m}^s(c_1, c_2) \!:=\! \big(\textstyle \frac{1}{m}\! \sum_{j=1}^m \!\!\frac{1}{|t_j|}\!\sum_k\! \mathds{1}_{O_k(c_1, t_j \vert e)=O_k(c_2, t_j \vert e)}\big)^s \!\text{, }
    \textstyle \widehat{p}_{n,m}( \mathcal N_c^s) \!:=\! \textstyle \frac{1}{n}\! \sum_{i=1}^n \!\widehat{\mathrm{sim}}_{p,e;m}^s(c, c_i).
\end{align*}
\end{definition}

Thus far, we have introduced Monte Carlo estimators of algorithmic code correctness $p(\mathcal N_c^s)$ for any sharpness coefficient $s \in \mathbb N$. But what is the benefit of smooth estimators? Since we assume $p$ to be parameterized by an LLM, measure smoothing is desirable for three key reasons: (1) it introduces a beneficial inductive bias that diminishes neural approximation artifacts, (2) it reduces the Monte Carlo estimator's bias, and (3) it significantly increases the estimator's signal-to-noise ratio (SNR).

\begin{theorem}[Measure smoothing inductive bias]\label{thm:measure_smoothing}
    Let $\mathrm{sim}_{p,e}^s(c,c^\prime) \geq 1-\varepsilon$ for $\varepsilon \in (0,1)$. Then 
    \begin{equation*}
    |p(\mathcal N_c^s) - p(\mathcal N_{c^\prime}^s)| \leq \sqrt{2\varepsilon}.
    \end{equation*}
\end{theorem}
Here, the sharpness exponent $s$ determines the granularity at which two solutions are considered similar and thus tunes the strength of the environment-dependent inductive bias.

\begin{theorem}[Consistency and bias]\label{thm:consistency_and_bias}
    $\widehat{\mathrm{sim}}_{p,e;m}^s$ is a strongly consistent estimator of $\mathrm{sim}_{p,e}^s$, i.e., $\mathbb P[\lim_{m \to \infty} \widehat{\mathrm{sim}}_{p,e;m}^s(c_1, c_2) = \mathrm{sim}_{p,e}^s (c_1, c_2)] = 1\ \forall c_1, c_2 \in \mathcal V^{\le L}$. Yet, it is biased unless $s=1$.
\end{theorem}
Here, the bias arises from Jensen's inequality combined with the unbiasedness of $\widehat{\mathrm{sim}}_{p,e;m}^1$.

\begin{theorem}[SNR dominance]\label{thm:snr_dominance}
    Let $\mathrm{sim}_{p,e}^1(c_1, c_2) = \mu \in (0,1)$ be the true functional similarity.
We define the signal-to-noise Ratio (SNR) of an estimator $\widehat{X}$ as $\mathrm{SNR}(\widehat{X}) := (\mathbb{E}[\widehat{X}])^2 / \mathrm{Var}(\widehat{X})$.
    For any number of test suites $m \ge 1$, the SNR of the soft estimator dominates the hard estimator by a factor of at least $m^2$:
    \begin{equation*}
        {\mathrm{SNR}(\widehat{\mathrm{sim}}_{p,e;m}^1(c_1, c_2))}\big/{\mathrm{SNR}(\widehat{\mathrm{sim}}_{p,e;m}^\infty(c_1, c_2))} \geq m (1/{\mu})^{m-1} \cdot \tfrac{1-\mu^m}{1-\mu} \ge m^2
    \end{equation*}
\end{theorem}
Theorem~\ref{thm:snr_dominance} contains a uniform lower bound of $m^2$, and the much stronger $\mu$-dependent lower bound in $\Omega(m(1/\mu)^{m})$. Thus, the SNR dominance is strongest for small $\mu$, i.e., for diverse generations.

\subsection{A dual expression for code correctness}
\label{sec:primal_dual}

Revisit the generative process of Figure~\ref{fig:probabilistic_model}. Given an environment $e \in \mathcal E$ and an executable specification $t \in \mathcal V^{\le L}$, a correct implementation $c \in \mathcal V^{\le L}$ is such that $R(c,t\,\vert\, e) = 1$. Thus, code correctness can be estimated via $p(\mathcal N_c^\infty)$, or, alternatively, via $\mathbb E_{t\sim p}[R(c,t\, \vert \, e)^\infty]$. Both are equally valid a priori.
\begin{definition}[Code-test calibration]\label{def:code_test_calibration}
    A model $p(c,t)$ is code-test calibrated in an environment $e \in \mathcal E$ if 
    \begin{equation*}
        p(\mathcal N_c^\infty) = \mathbb P_{t \sim p}[R(c,t\, \vert \, e) = 1].
    \end{equation*}
\end{definition}
Note that if a model is code-test calibrated, we expect equal downstream performance with both the primal and dual expressions. Yet, as our experiments in Section~\ref{sec:experiments_post_generation_selection_heuristics} demonstrate, current frontier models are miscalibrated in practice. In analogy to $p(\mathcal N_c^s)$, a smooth estimator $\mathbb E_{t\sim p}[R(c,t\,\vert\, e)^s]$ is also considered. The SNR dominance of Theorem~\ref{thm:snr_dominance} still applies, motivating the smooth variant.

\subsection{Popular selection 
heuristics as self-consistency}
Our framework allows us to identify several well-established selection heuristics as maximizing an estimate of the probability of correctness $p(\mathcal N_c^\infty)$. Table~\ref{tab:types_of_self_consistency} provides a summary of the code selection literature based on our formalism, with Appendix~\ref{sec:popular_heuristics} providing more information. In the experiments Section~\ref{sec:experiments}, we focus on boolean test outputs $\mathbb O$, i.e., the canonical case of test-driven development. 

\begin{table}
    \centering
    \vspace{-14pt}
    \caption{Selection heuristics maximize pass@1 using different estimators of code correctness. The test output $\mathbb O$ is binary $\mathbb B$ (pass/fail) or arbitrary $\mathbb A$. The asterisk marks that only single-test test suites were considered by \citet{lahiri2023interactivecodegenerationtestdriven, roziere2022leveraging, le2022coderl, chen2024b4}.}
    \label{tab:types_of_self_consistency}
    \setlength{\tabcolsep}{4pt}
    \resizebox{\linewidth}{!}{
    \begin{tabular}{lcl|lcl}
         & $\mathbb O$ & selection criterion &  & $\mathbb O$ & selection criterion \\
         \hline
         \textsc{MBR-Exec Hard} \citep{shi2022mbr} & $\mathbb B$ & $\hat{p}_{n,m}(\mathcal N^\infty_c)$ & \textsc{MBR-Exec Soft} \citep{shi2022mbr} &  $\mathbb B$ & $\hat{p}_{n,m}(\mathcal N^1_c)$\\
         \textsc{AlphaCode} \citep{li2022competition} &  $\mathbb A$ & $\hat{p}_{n,m}(\mathcal N^\infty_c)$ & \textsc{FunCoder} \citep{chen2024divideandconquer} &  $\mathbb A$ & $\hat{p}_{n,m}(\mathcal N^1_c)$\\
         \textsc{MaxPass Hard} (various works$^\ast$) &  $\mathbb B$ & $\hat{\mathbb E}_{t_{1:m}}[R^\infty]$ & \textsc{MaxPass Soft} (ours) &  $\mathbb B$ & $\hat{\mathbb E}_{t_{1:m}}[R]$\\
         \textsc{CodeT Hard} \citep{chen2023codet} & $\mathbb B$ & $\hat{p}_{n,m}(\mathcal N^\infty_c) \hat{\mathbb E}_{t_{1:m}}[R]$ & \textsc{CodeT Soft} (ours) &  $\mathbb B$ & $\hat{p}_{n,m}(\mathcal N^1_c) \hat{\mathbb E}_{t_{1:m}}[R]$\\
         \hline
    \end{tabular}
    }
\end{table}

\subsection{Maximizing pass@k via probability of correctness}
Our analysis focuses on robust estimators of code correctness, i.e., pass@1. Nevertheless, as Proposition~\ref{prop:additivity_probability_of_correctness} shows, such estimators can also be utilized to directly maximize pass@k.
\begin{proposition}[Additivity of probability of correctness]\label{prop:additivity_probability_of_correctness}
    Suppose that the true test suite $t \in \mathcal V^{\le L}$ verifying an implementation $c \in \mathcal V^{\le L}$ in environment $e\in \mathcal E$ is almost surely such that $\exists !$ correct $\mathcal N_c^\infty$. Define $\mathrm{pass@k}(\{\mathcal N_{c_i}^\infty\}_{i=1}^k) := \mathbb P_{t\sim p} [\exists i \in \{1,\ldots,k\} : R(c_i,t\, \vert\, e)=1]$. Then, $\mathrm{pass@k}$ is maximized by greedily selecting the $k$ code classes $\{\mathcal N_{c_i}^\infty\}_{i=1}^k$ with largest $\mathrm{pass@1}$.
\end{proposition}

\section{Environment feedback during generation}\label{sec:feedback_during_generation}
Having established a unifying framework for environment-based post-generation selection, we now turn to environment feedback \textit{during} generation. Recall from Figure~\ref{fig:probabilistic_model} that the executable specification (test suite) and the algorithm implementation depend causally on the latent abstract algorithm $a \in \mathcal A$ and the execution environment $e \in \mathcal E$. Whereas the abstract algorithm remains unobservable beyond the task description $d \in \mathcal V^{\le L}$, the environment can be probed to retrieve properties like \textit{syntax rules} and \textit{library behaviors} via execution results $R(c, t\,\vert\, e) \in [0,1]$ or descriptive reports $U(c, t\,\vert\, e) \in \mathcal V^{\le L}$. 

\subsection{Backprompting approximates Thompson sampling}
To permit theoretical analysis of backprompting, we model the iterative refinement loop as a sequential bandit problem where the LLM acts as an in-context conditioned policy. We define an action $x := (c, t) \in \mathcal V^{\le L} \times \mathcal V^{\le L}$ as a pair consisting of an implementation and a test suite. Thompson sampling then selects an action $x$ according to the probability that it is the optimal solution $x^*$ given the task $d$ and observation history $\mathcal{H}_n$, i.e., its policy is $\pi(x\,\vert\,\mathcal{H}_n, d) := \mathbb P[x= x^* \,\vert\, \mathcal{H}_n, d]$.

We argue that LLMs naturally learn to approximate this policy. Pre-training data for code often contains description-attempt-solution triplets $(d, \mathcal{H}, x^*)$---for example, software repositories where commit histories document failed attempts followed by a verified solution. Following \citeauthor{de2025simplifying} \citeyearpar{de2025simplifying}, training a sequence model on such trajectories with a cross-entropy loss forces the model to minimize the Kullback-Leibler divergence to the true conditional distribution of the optimum:
\begin{align*}
    \textstyle \theta^* := & \arg\!\min_\theta \mathbb{E}_{(d, \mathcal{H}, x^*) \sim p_{\text{data}}} \left[ -\log p_\theta(x^* \,\vert\, \mathcal{H}, d) \right]\\
    = & \textstyle \arg\!\min_\theta \mathbb{E}_{(d, \mathcal{H}) \sim p_{\text{data}}}\! \left[ D_{KL}(p_{\text{data}}(x^*\vert \mathcal{H}, d) \vert\vert p_\theta(x^* \vert \mathcal{H}, d)) \right]\!.
\end{align*}
Hence, by backprompting failed attempts during generation, we sample $x_{n+1} \sim p_\theta(\, \cdot \,\vert\, \mathcal{H}_n, d)$, which, due to pre-training, results in drawing from the model's learned belief over the optimal solution $x^*$. \begin{rebuttal}To empirically verify this approximation, Appendix~\ref{sec:empirical_signatures_ts} provides indirect evidence of the LLM's sampling distribution: the generated programs \& test suites remain diverse yet stable across rounds, with the model's implicit posterior belief naturally contracting as execution feedback accumulates.\end{rebuttal}

This equivalence allows us to lift regret bounds from Bayesian optimization to our setting. We establish a novel regret bound for Thompson sampling which accounts for unobservable reward components, reflecting irreducible under-specification of the abstract algorithm in the task description.

\subsection{A Bayesian regret bound for Thompson sampling with unobservable reward components}\label{sec:regret_bounds}
Suppose a reward function $r(x) = r_{obs}(x) + r_{hid}(x)$. \begin{rebuttal}For us, $r(c,t) = 1$ if and only if $c$ correctly implements the algorithm, i.e., if $R(c, t_{true} \, \vert \, e) = 1$. \end{rebuttal} Now, $r_{obs}(c,t) = R(c,t \, \vert \, e)$ is an observable consistency score between codes, tests, and environments. In contrast, $r_{hid}(c,t)$ \begin{rebuttal} is a corrective factor $r(c)-r_{obs}(c,t)$ that captures \end{rebuttal}the task relevance of $(c,t)$ to the unobservable latent algorithm $a \in \mathcal A$, partially described by $d \in \mathcal V^{\le L}$. Note that test pass rates $r_{obs}$ are bounded in $[0,1]$ and thus inherently subgaussian by Hoeffding's Lemma. \begin{rebuttal}The same holds for $r_{hid} \in [-1, 1]$ as the algebraic residual.\end{rebuttal} \looseness-1

We consider the regret against the true optimal solution $x^* := \arg\max_x r(x)$ and against Thompson sampling with full knowledge of the environment, i.e., against $x^*_{obs} \sim \mathbb P[x^* \,\vert\, r_{obs}]$.

\begin{minipage}{\linewidth}
\begin{theorem}\label{thm:reward_bound}
Decompose the reward into independent components $r := r_{obs} + r_{hid}$, where $\mathbb E[r_{hid}]$ is finite, and $r_{obs}(x)$ is $\sigma_n(x)$-subgaussian given observation history $\mathcal H_n := (x_i, r_{obs}(x_i))_{i=1}^{n-1}$. Define the Thompson sampling policies $x_{obs}^* \sim \mathbb P[x^* \,\vert\, r_{obs}]$ and $x_n \sim \mathbb P[x^* \,\vert\, \mathcal H_n]$ for $x^* := \arg\max_x r(x)$. Then
\begin{equation*}
    \textstyle \mathbb E[\sum_{n=1}^T r(x^*_{obs}) - r(x_n)] \leq \beta\sqrt{T \cdot \Gamma_T}
\end{equation*}
for $\beta := \sqrt{2 \log(|\mathrm{dom}(r)|)} + 3$ and cumulative uncertainty $\Gamma_T := \mathbb E[\sum_{n=1}^T\!\sigma_n^2(x_n)]$.
\end{theorem}
\end{minipage}

\paragraph{Interpretation} If $r_{obs}$ has independent entries, then the cumulative uncertainty $\Gamma_T$ is upper bounded by $\sum_{x \in \mathrm{dom}(r)} \sigma_1^2(x)$. Conversely, if $r_{obs}$ has correlated entries, we expect faster convergence. Indeed, for correlated strictly subgaussian rewards, observing $r_{obs}(x_n)$ reduces uncertainty at $x \not = x_n$ in expectation (see Lemma~\ref{lem:strict_expected_subgaussian_variance}). For Gaussian $r_{obs}$, we can formalize the benefit of correlation information in closed form: according to Lemma~\ref{lem:information_gain_for_gaussians_noiseless_gamma}, $\Gamma_T$ is upper bounded by the maximal information gain~\citep{srinivas2009gaussian}, which reflects the true degrees of freedom of the function $r_{obs}$. E.g., if $r_{obs}$ is a linear function of $D$ features (with standard Gaussian weights), $\Gamma_T \in \mathcal O(D \log T)$.\looseness-1

The regret in Theorem~\ref{thm:reward_bound} is stated with respect to a Thompson sampling policy with full knowledge of $r_{obs}$. The regret with respect to the ground-truth optimal solution follows:

\begin{corollary}\label{cor:reward_bound}
    Consider the setting of Theorem~\ref{thm:reward_bound}. Then for $\Delta := \mathbb E[r(x^*) - r(x^*_{obs})]$ it holds that
    \begin{equation*}
        \textstyle \mathbb E[\sum_{n=1}^T r(x^*) - r(x_n)]\leq \beta\sqrt{T \cdot \Gamma_T} + T \Delta.
    \end{equation*}
\end{corollary}

\paragraph{Interpretation} The bound decomposes cumulative regret into two distinct terms. The first term, ${\mathcal{O}}(\sqrt{T})$, is sub-linear and represents the \textit{learnable} part of the problem: by interacting with the environment (running tests on implementations), the assistant reduces its uncertainty about $r_{obs}$ (syntax, library behavior, etc.). The second term, $T\Delta$, is linear. It represents the \textit{irreducible regret}: even if the assistant perfectly understands the execution environment ($r_{obs}$), the informal description $d$ may still be ambiguous regarding the true latent algorithm $a$. No amount of execution feedback can resolve this specific ambiguity. \begin{rebuttal}$\Delta$ can be directly estimated by measuring the performance gap at convergence between an oracle setting (where $t_{true}$ is provided such that $r_{hid} = 0$ and thus $x_{obs}^* = x^*$) and a self-test setting. In practice, the oracle is typically unavailable. As an alternative, Section~\ref{sec:estimating_irreducible_regret} provides a simple estimator for $\Delta$ via the collision probability of functional equivalence.\end{rebuttal}

\paragraph{Predicting reward} Given tight regret bounds, a saturated $\Gamma_T \approx \Gamma$, and with $\sqrt{T} \text{-} \sqrt{T\text{-}1} \approx T^{\, \text{-}1/2}/2$, we get the behavior
\begin{equation}\label{eq:predicting_reward_curves}
    \mathbb E[r(x_T)] \approx \mathbb E[r(x^*)] -  \Delta - T^{-1/2}\cdot \beta \cdot \sqrt{\Gamma}/ 2.
\end{equation}
We thus expect in-context Thompson sampling to follow a reward curve that approaches a ceiling determined by $\Delta$ at a rate of $\Theta(1/\sqrt{T})$. As Figure~\ref{fig:in_context_thompson_sampling_line_plot} demonstrates, this indeed occurs if the irreducible regret $\Delta$ is small and does not dominate the learnable part (e.g., purple lines across all three datasets).

\begin{rebuttal}

\subsection{Estimating irreducible regret during backprompting}\label{sec:estimating_irreducible_regret}
Suppose $r(c) \in \{0, 1\}$ is $1$ if and only if $R(c, t_{true} \, \vert \, e) = 1$. Then, the irreducible regret $\Delta$ satisfies
\begin{align*}
    1- \Delta & = 1 - \mathbb E[r(x^*) - r(x^*_{obs})] = \mathbb E[r(x^*_{obs})] = \mathbb P[R(c_{obs}^*, t_{true} \, \vert \, e) = 1]\\
    & = \mathbb E_{c^*_{obs}}[\mathbb E_{t_{true}}[\mathds{1}_{R(c^*_{obs},t_{true} \, \vert \, e) = 1}]] = \mathbb E_{c^*_{obs}}[p(\mathcal N_{c^*_{obs}}^\infty)] = \mathbb E_{c^*_{obs,1}, c^*_{obs,2}}[\mathrm{sim}_{p,e}^\infty(c^*_{obs,1}, c^*_{obs,2})].
\end{align*}
All expectations are understood conditioned on full knowledge of $r_{obs}$. Note that the expression for $1-\Delta$ corresponds to the collision probability of the equivalence classes $\mathcal N_c^\infty$, whose Monte Carlo estimator can be derived from Section~\ref{sec:similarity_vs_equivalence}. To sample $x_{obs}^*$, a moving window $(x_n)_{n=T-w+1}^T$ of size $w$ may be used upon convergence of the Thompson sampling policy $\pi_n(x) = \mathbb P[x=x^* \, \vert \, \mathcal H_n]$. If $\Delta$ exceeds a user-defined threshold, clarifications to resolve the task ambiguity can be requested.\looseness-1 
\end{rebuttal}

\section{Experiments}\label{sec:experiments}
To corroborate our theory, we empirically compare the selection heuristics unified in Section~\ref{sec:post_generation_selection_via_feedback}, and confirm the fundamental limits of in-the-loop backprompting of environment feedback, established theoretically in Section~\ref{sec:feedback_during_generation} via identification of backprompting with approximate in-context Thompson sampling. Of particular interest is the impact of the sharpness factor $s$ and irreducible regret $\Delta$.

\subsection{Setup}
\paragraph{Datasets}
We consider three challenging python competitive coding datasets. \textit{BigCodeBenchHard} is the partition of the 148/1140 most difficult tasks in BigCodeBench~\citep{zhuo2024bigcodebench}. It adopts the format of HumanEval~\citep{chen2021evaluating}, but covers interaction with 139 external libraries.
\textit{QiskitHumanEvalSim} is a novel subset of the 143/151 QiskitHumanEval~\citep{vishwakarma2024qiskit} questions that can be solved without access to a real-world quantum computer, democratizing the dataset. Finally, \textit{LeetCodeDataset} consists of 228 diverse puzzle-style questions from LeetCode released after 2024-07-01 and collected by \citet{xia2025leetcodedatasettemporaldatasetrobust}. In contrast to the other datasets, LeetCodeDataset typically includes example input-output pairs in natural language, significantly mitigating the \text{irreducible regret} $\Delta$ identified in Corollary~\ref{cor:reward_bound}. \begin{rebuttal}Lastly, we consider agentic coding with \textit{SWE-bench Verified}~\citep{jimenez2024swebench}, a benchmark where LLM agents must navigate Github code repositories and resolve real-world software issues\end{rebuttal}. For more information, see Appendix~\ref{sec:datasets}.

\paragraph{Models}
We consider \begin{rebuttal}seven\end{rebuttal} open-weight LLMs: Qwen3-235B-A22B-Instruct-2507~\citep{yang2025qwen3technicalreport}, the hybrid Transformer-SSM models: Qwen3-Next, Qwen3-Coder-Next, and \begin{rebuttal}Qwen3.5-397B-A17B\end{rebuttal}~\citep{qwen3codernexttechnicalreport}, GPT-OSS-120B~\citep{agarwal2025gpt}, and MiniMax-M2.1 as well as MiniMax-M2.7~\citep{chen2025minimax}. Whereas the Qwen3 models respond immediately, the others (including Qwen3.5) are chain-of-thought reasoning models.
Only MiniMax-M2.1\begin{rebuttal}, MiniMax-M2.7,\end{rebuttal} and Qwen3-Coder-Next are dedicated coding models. More information is in Appendix~\ref{sec:extensive_model_description}. 

\subsection{Post-generation selection heuristics}\label{sec:experiments_post_generation_selection_heuristics}
Consider Table~\ref{tab:post_generation_selection_methods_results} (and its more detailed version in Figure~\ref{fig:post_generation_selection_heuristics_line_plot} of Appendix~\ref{sec:post_selection_appendix}), reporting on the accuracy of three models (Qwen3-235B-A22B-Instruct-2507, GPT-OSS-120B, MiniMax-M2.1) achieved through selection heuristics applied after $10$ independent code and test suite generations. The accuracy improvements over \textsc{Random Selection} are highly sensitive to the post-generation selection method, with significant gains achieved by \textsc{MaxPass} and \textsc{CodeT} strategies. In particular, note that soft estimators generally outperform hard estimators. This corroborates the theoretical insights of Section~\ref{sec:similarity_vs_equivalence}: soft estimators of probability of correctness benefit from a smoothness inductive bias and from an improved SNR. \begin{rebuttal}In Appendix~\ref{sec:rank_correlation_snr}, we isolate this effect explicitly: while hard estimators exhibit slightly better rank correlation with oracle correctness, soft estimators achieve vastly higher empirical SNR, directly mediating the observed pass@1 improvements.\end{rebuttal}
\begin{table*}[b]
    \centering
    \caption{Pass@1 (\%) of post-generation selection heuristics (see Table~\ref{tab:types_of_self_consistency}) after 10 rounds of independent code and test generation ($m=n=10$).
    Within a family of heuristics, we bold the more accurate one. The soft variants generally outperform the hard ones, corroborating the insights of Theorem~\ref{thm:measure_smoothing}, Theorem~\ref{thm:consistency_and_bias}, and Theorem~\ref{thm:snr_dominance}. We report mean ($\pm$ standard error) over 5 seeds.}
    \label{tab:post_generation_selection_methods_results}
    \small
    \resizebox{\linewidth}{!}{
    \setlength{\tabcolsep}{2pt}
    \begin{tabular}{lccccccccc}
        & \multicolumn{3}{c}{BigCodeBenchHard (BigCode)} & \multicolumn{3}{c}{QiskitHumanEvalSim (Qiskit)} & \multicolumn{3}{c}{LeetCodeDataset (LeetCode)}\\
        \cmidrule(lr){2-4}\cmidrule(lr){5-7}\cmidrule(lr){8-10}
        & Qwen3 & GPT-OSS & MM-M2.1 & Qwen3 & GPT-OSS & MM-M2.1 & Qwen3 & GPT-OSS & MM-M2.1\\
        \cmidrule(r){0-0}\cmidrule(lr){2-4}\cmidrule(lr){5-7}\cmidrule(lr){8-10}
\textsc{Random Selection} & $27.43^{\scriptscriptstyle\pm 0.31}$ & $37.16^{\scriptscriptstyle\pm 0.63}$ & $34.46^{\scriptscriptstyle\pm 0.45}$ & $52.73^{\scriptscriptstyle\pm 0.19}$ & $41.96^{\scriptscriptstyle\pm 0.36}$ & $43.78^{\scriptscriptstyle\pm 0.47}$ & $18.60^{\scriptscriptstyle\pm 0.22}$ & $69.12^{\scriptscriptstyle\pm 0.31}$ & $62.54^{\scriptscriptstyle\pm 0.58}$ \\
\cmidrule(r){0-0}\cmidrule(lr){2-4}\cmidrule(lr){5-7}\cmidrule(lr){8-10}
\textsc{MBR-Exec Hard} & $29.32^{\scriptscriptstyle\pm 0.43}$ & $37.30^{\scriptscriptstyle\pm 0.42}$ & $35.95^{\scriptscriptstyle\pm 0.55}$ & $54.97^{\scriptscriptstyle\pm 0.13}$ & $\mathbf{47.69}^{\scriptscriptstyle\pm 0.40}$ & $45.73^{\scriptscriptstyle\pm 0.34}$ & $18.33^{\scriptscriptstyle\pm 0.18}$ & $\mathbf{72.81}^{\scriptscriptstyle\pm 0.41}$ & $73.07^{\scriptscriptstyle\pm 0.33}$ \\
\textsc{MBR-Exec Soft} & $\mathbf{29.73}^{\scriptscriptstyle\pm 0.42}$ & $37.30^{\scriptscriptstyle\pm 0.54}$ & $\mathbf{36.76}^{\scriptscriptstyle\pm 0.31}$ & $54.97^{\scriptscriptstyle\pm 0.16}$ & $47.13^{\scriptscriptstyle\pm 0.35}$ & $\mathbf{46.71}^{\scriptscriptstyle\pm 0.23}$ & $\mathbf{18.60}^{\scriptscriptstyle\pm 0.16}$ & $72.19^{\scriptscriptstyle\pm 0.32}$ & $\mathbf{73.33}^{\scriptscriptstyle\pm 0.45}$ \\
\cmidrule(r){0-0}\cmidrule(lr){2-4}\cmidrule(lr){5-7}\cmidrule(lr){8-10}
\textsc{CodeT Hard} & $31.22^{\scriptscriptstyle\pm 0.32}$ & $37.30^{\scriptscriptstyle\pm 0.24}$ & $36.08^{\scriptscriptstyle\pm 0.59}$ & $59.02^{\scriptscriptstyle\pm 0.19}$ & $51.33^{\scriptscriptstyle\pm 0.38}$ & $50.49^{\scriptscriptstyle\pm 0.23}$ & $21.23^{\scriptscriptstyle\pm 0.12}$ & $73.25^{\scriptscriptstyle\pm 0.36}$ & $78.60^{\scriptscriptstyle\pm 0.24}$ \\
\textsc{CodeT Soft} (ours) & $\mathbf{35.00}^{\scriptscriptstyle\pm 0.36}$ & $\mathbf{38.24}^{\scriptscriptstyle\pm 0.34}$ & $\mathbf{42.03}^{\scriptscriptstyle\pm 0.42}$ & $\mathbf{62.24}^{\scriptscriptstyle\pm 0.28}$ & $\mathbf{54.13}^{\scriptscriptstyle\pm 0.27}$ & $\mathbf{53.85}^{\scriptscriptstyle\pm 0.33}$ & $\mathbf{24.30}^{\scriptscriptstyle\pm 0.23}$ & $\mathbf{74.56}^{\scriptscriptstyle\pm 0.49}$ & $\underline{\mathbf{79.56}}^{\scriptscriptstyle\pm 0.24}$ \\
\cmidrule(r){0-0}\cmidrule(lr){2-4}\cmidrule(lr){5-7}\cmidrule(lr){8-10}
\textsc{MaxPass Hard} & $33.51^{\scriptscriptstyle\pm 0.45}$ & $37.30^{\scriptscriptstyle\pm 0.46}$ & $39.59^{\scriptscriptstyle\pm 0.39}$ & $61.40^{\scriptscriptstyle\pm 0.15}$ & $48.53^{\scriptscriptstyle\pm 0.40}$ & $51.75^{\scriptscriptstyle\pm 0.22}$ & $24.82^{\scriptscriptstyle\pm 0.26}$ & $70.09^{\scriptscriptstyle\pm 0.54}$ & $76.32^{\scriptscriptstyle\pm 0.50}$ \\
\textsc{MaxPass Soft} (ours) & $\mathbf{35.41}^{\scriptscriptstyle\pm 0.38}$ & $\mathbf{38.24}^{\scriptscriptstyle\pm 0.34}$ & $\underline{\mathbf{42.30}}^{\scriptscriptstyle\pm 0.31}$ & $\underline{\mathbf{63.22}}^{\scriptscriptstyle\pm 0.19}$ & $\mathbf{54.13}^{\scriptscriptstyle\pm 0.27}$ & $\underline{\mathbf{54.13}}^{\scriptscriptstyle\pm 0.34}$ & $\mathbf{25.09}^{\scriptscriptstyle\pm 0.29}$ & $\mathbf{74.56}^{\scriptscriptstyle\pm 0.49}$ & $\mathbf{79.21}^{\scriptscriptstyle\pm 0.30}$ \\
\bottomrule
\end{tabular}}
\end{table*}
Table~\ref{tab:post_generation_selection_methods_results} also reveals that \textsc{MaxPass} consistently outperforms \textsc{MBR-Exec} across datasets and models. This performance gap provides empirical evidence that current frontier models are not perfectly code-test calibrated. Because the theoretical duality between code-based and test-based estimation breaks down in practice, our results indicate that the test-based dual estimation of code correctness (\textsc{MaxPass}) is the superior, more robust heuristic for current uncalibrated models given binary test outputs of test-driven development. \begin{rebuttal}Finally, note that Appendix~\ref{sec:b4_comparison} demonstrates that \textsc{MaxPass Soft} outperforms B4 \citep{chen2024b4}, a recent Bayesian selection baseline relying on strict equivalence.\end{rebuttal}

\subsection{Approximate in-context Thompson sampling via iterative backprompting}\label{sec:in_context_thompson_sampling}
Next, let us consider conditioning the LLM during generation on execution results from running a potential implementation $c$ on a potential test suite $t$. Based on the development in Section~\ref{sec:feedback_during_generation}, we provide the LLM in-context with an informative test report $U(c,t\,\vert\, e) \in \mathcal V^{\le L}$. Since test reports can become very long, direct backprompting of all information from multiple rounds quickly becomes computationally infeasible. Thus, we compress $U(c,t\,\vert\, e)$ into a much shorter token sequence using the LLM. We ablate on the compression technique and the factorization of $p(c,t)$ in Appendix~\ref{sec:ablation}. This section conducts the main experiments using Qwen3-235B-A22B-Instruct-2507, with findings on other models, such as their low test quality (GPT-OSS-120B and MiniMax-M2.1) and positive generalization across architectures within the Qwen family (Qwen3-Next and Qwen3-Coder-Next, Qwen3.5-397B-A17B), discussed at the end. Now, examine Figure~\ref{fig:in_context_thompson_sampling_line_plot}, where we compare three settings:\looseness-1

\paragraph{Pass@1 with \textsc{Oracle-Test Feedback}} First, we consider the simplified setting where the test suite $t_{oracle}$ is given and the model only has to find a code $c$ such that $R(c,t_{oracle}\,\vert\, e) = 1$. In the formalism of Theorem~\ref{thm:reward_bound}, this corresponds to the setting where $r_{hid}$ is known, which according to Corollary~\ref{cor:reward_bound} permits global convergence to $x^*$. We observe that under these circumstances, the LLM acts as a highly effective approximate Thompson sampler, finding solutions that are much better than what can be achieved with post-selection heuristics. Indeed, in-context Thompson sampling achieves a pass@1 of $63.78 \%$, $72.31\%$, and $72.11\%$ across BigCodeBenchHard, QiskitHumanEvalSim, and LeetCodeDataset, notably beating pass@10 of unconditioned generation: 
$46.35\%$, $67.83\%$, $28.60 \%$. To tie into our theory, we may use Equation~\eqref{eq:predicting_reward_curves} (with $T=1,10$) to estimate the rate of information gain $\beta \cdot \sqrt{\Gamma} / 2$. We obtain the values $42.69\%$ (BigCode), $32.12\%$ (Qiskit), and $64.66\%$ (LeetCode).

\begin{figure*}[b]
    \centering
    \includegraphics[width=\linewidth]{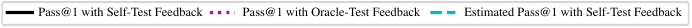}
    
    \begin{subfigure}[b]{0.33\linewidth}
        \centering
        \captionsetup{justification=centering}\includegraphics[width=\linewidth]{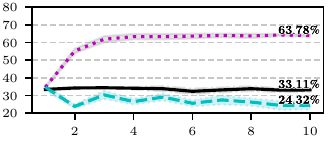}
        \caption{BigCodeBenchHard}
    \end{subfigure}%
    \hfill
    \begin{subfigure}[b]{0.33\linewidth}
        \centering
        \captionsetup{justification=centering}\includegraphics[width=\linewidth]{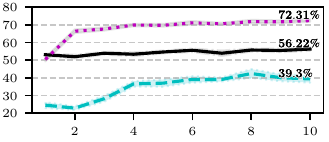}
        \caption{QiskitHumanEvalSim}
    \end{subfigure}%
    \hfill
    \begin{subfigure}[b]{0.33\linewidth}
        \centering
        \captionsetup{justification=centering}\includegraphics[width=\linewidth]{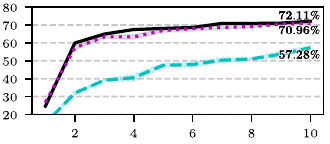}
        \caption{LeetCodeDataset}
    \end{subfigure}
    
    \caption{Pass@1 (\%) during $10$ rounds of in-context Thompson sampling using Qwen3-235B-A22B-Instruct-2507. We report similar trends across other members of the Qwen3 family, including hybrids, in Appendix~\ref{sec:hybrid}. W.r.t. Theorem~\ref{thm:reward_bound}, the black solid line denotes the true reward $r$, and the cyan dashed line $r_{obs}$. The purple dotted line is the simplified setting of known $r_{hid}$, which, according to Corollary~\ref{cor:reward_bound}, permits global convergence to $x^*$. We report mean ($\pm$ standard error) over 5 seeds.}
    \label{fig:in_context_thompson_sampling_line_plot}
\end{figure*}

\paragraph{Estimated pass@1 with \textsc{Self-Test Feedback}} Next, we let the language model generate both codes $c$ and test suites $t$, but report on $R(c, t\, \vert\, e)$ instead of $R(c, t_{oracle}\, \vert\, e)$. Note that in the formalism of Theorem~\ref{thm:reward_bound}, the language model is asked to optimize $r = r_{obs} + r_{hid}$, but we only report on the trajectory of $r_{obs}(x_n)$. Although the LLM successfully optimizes $r_{obs}$ on QiskitHumanEvalSim and LeetCodeDataset, it struggles to do so for BigCodeBenchHard, demonstrating that joint optimization of code and tests proves more challenging than only optimizing code, even for observable objectives.

\paragraph{Pass@1 with \textsc{Self-Test Feedback}} Finally, we consider the most challenging setting of optimizing $R(c, t_{oracle}\, \vert\, e)$ with only access to $R(c, t\, \vert\, e)$. Recall Corollary~\ref{cor:reward_bound}, which predicts a small uncertainty on $r_{hid}$ required for success, i.e., the task description must contain sufficient information on the algorithmic behavior. Whereas on LeetCodeDataset the assistant recovers almost the same trajectory as if \textsc{Oracle-Test Feedback} was present, the transfer on QiskitHumanEvalSim is more modest and completely missing on BigCodeBenchHard. Note that our theory actually predicts these findings: the task description on LeetCodeDataset includes input-output examples, decreasing the uncertainty of $r_{hid}$ and thus the irreducible regret $\Delta$. In fact, the difference in final performance between \textsc{Self-Test Feedback} and \textsc{Oracle-Test Feedback} can be used as an estimator for the irreducible regret $\Delta$ of each dataset, giving 30.68\%, 16.08\%, and \begin{rebuttal}$\approx$ 0\%\end{rebuttal}, for BigCode, Qiskit, and LeetCode, respectively. \looseness -1

\paragraph{The importance of capturing the test manifold}
Appendix~\ref{sec:backprompt_gpt_mm} evaluates GPT-OSS-120B and MiniMax-M2.1 as in-context Thompson samplers. While they are highly effective optimizers under \textsc{Oracle-Test Feedback}, both models show almost no improvement when backprompting \textsc{Self-Test Feedback}. Moreover, the gap between self-estimated pass@1 and true pass@1 is substantially larger for these two models, indicating poorer modeling of the test manifold. \begin{rebuttal}In contrast, our primary model effectively captures this manifold. As detailed in Appendix~\ref{sec:quality_and_discriminative_power_of_generated_tests}, the tests generated by Qwen3 provide a meaningful, discriminative evaluation signal with a typical bimodal spectrum of difficulty.\end{rebuttal}

\paragraph{Generalization within Qwen family} 
Motivated by Qwen3-235B-A22B's strong performance, we evaluate the hybrid 80B-A3B parameter models Qwen3-Next and Qwen3-Coder-Next \begin{rebuttal}and the hybrid Qwen3.5-397B-A17B\end{rebuttal} (Appendix~\ref{sec:hybrid}). The strong self-test backprompting capabilities within the architecturally diverse Qwen3 model family suggests that neither model architecture nor size substantially limit the effectiveness of backprompting and instead hints at training as a potential limiting factor. This highlights a critical lesson for autonomous software engineering: to unlock the benefits of in-the-loop environment interaction, writing tests is just as important as writing code. \looseness -1

\subsection{QiskitHumanEvalSimX}\label{sec:qiskit_human_eval_sim_X}
The central lesson of Section~\ref{sec:in_context_thompson_sampling} is that the task description $d$ must contain sufficient information on the abstract algorithm $a$ such that the irreducible regret $\Delta$ does not dominate the total regret of in-context Thompson sampling. An effective strategy to improve algorithmic specification is to add input-output pairs to the task description. While LeetCodeDataset already contains such behavioral examples in the prompt, explaining the success of in-context Thompson sampling thereon, BigCodeBenchHard and QiskitHumanEvalSim do not. To further
\begin{wrapfigure}{r}{0.5\textwidth}
    \centering
    \vspace{-8pt}\includegraphics[width=\linewidth]{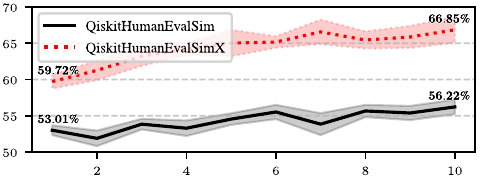}
    \caption{Pass@1 (\%) during $10$ rounds of backprompting on QiskitHumanEvalSimX. We use a Qwen3-235B-A22B-Instruct-2507. Results are reported as mean ($\pm$ standard error) over 5 seeds.}
    \label{fig:qiskit_human_eval_X}
    \vspace{-20pt}
\end{wrapfigure}
corroborate our theory, we introduce a new dataset: 
\textsc{QiskitHumanEvalSim with e\textbf{X}amples}. This new dataset builds on QiskitHumanEvalSim by querying Gemini 3 Pro to add a natural language description of the oracle test to the task description, including input-output examples. The exact prompt is listed in Appendix~\ref{sec:creating_qiskit_human_eval_simX}. As can be observed in Figure~\ref{fig:qiskit_human_eval_X}, in-context Thompson sampling indeed becomes more effective given the additional test description: instead of an accuracy delta over the first round of $3.21 \%$ we now report the delta $7.13\%$.

\begin{rebuttal}
\subsection{Estimating irreducible regret during in-context Thompson sampling}

\begin{wraptable}{r}{0.53\textwidth}
    \centering
    \vspace{-13pt}
    \caption{\begin{rebuttal}Online estimation of the irreducible regret $\Delta$ (\%) on Qwen3-235B-A22B-Instruct-2507 trajectories. We report mean ($\pm$ standard error) over 5 seeds.\end{rebuttal}}
    \label{tab:estimating_irreducible_regret}
    \small
    \setlength{\tabcolsep}{5pt}
    \begin{tabular}{lcc}
        Method  & BigCode & Qiskit\\
        \cmidrule(r){1-1} \cmidrule(lr){2-2}\cmidrule(lr){3-3}
        \textsc{Oracle Estimator}& $30.68^{\scriptscriptstyle\pm 1.69}$ & $16.08^{\scriptscriptstyle\pm 0.54}$ \\[0.3ex]
        \textsc{Online Estimator} &
        $27.07^{\scriptscriptstyle\pm 0.79}$ & $11.66^{\scriptscriptstyle\pm 1.02}$ \\
        \bottomrule
    \end{tabular}
    \vspace{-\intextsep}
\end{wraptable}

Having demonstrated in Section~\ref{sec:in_context_thompson_sampling} and Section~\ref{sec:qiskit_human_eval_sim_X} that effective in-context Thompson sampling requires unambiguous task descriptions, we investigate if the irreducible regret $\Delta$ can be reliably estimated online, e.g., to trigger a task clarification from the user. As shown in Table~\ref{tab:estimating_irreducible_regret}, a Monte Carlo estimator of $\Delta = 1-\mathbb E_{c^*_{obs,1}, c^*_{obs,2}}[\mathrm{sim}_{p,e}^\infty(c^*_{obs,1}, c^*_{obs,2})]$, developed in Section~\ref{sec:estimating_irreducible_regret} (\textsc{Online Estimator}, $w=3$), roughly captures the oracle-based estimation of the irreducible regret $\Delta = \mathbb E[r(x_T) - r(x_{T}^{oracle})]$ from Section~\ref{sec:in_context_thompson_sampling} (\textsc{ORACLE ESTIMATOR}). Here, we omit LeetCodeDataset since the Thompson sampling policy $\pi_n$ has not yet converged (to $\mathbb P[x^* \, \vert \, r_{obs}]$) after $10$ rounds, as evidenced by the non-saturated $r_{obs}$ in Figure~\ref{fig:in_context_thompson_sampling_line_plot} (cyan dashed line).

\end{rebuttal}

\subsection{\rebuttal{Extension to multi-step agentic code generation}}
Our central investigation deliberately isolates core components of coding agents~\citep{yang2024swe}---post-generation selection and iterative revision via environment feedback---within the scope of competitive programming. This focus ensures controlled experiments with well-defined environment interactions, avoiding the confounding variables inherent in multi-file editing workflows. However, the theoretical lessons and findings presented above naturally extend to fully agentic systems. For instance, agentic benchmarks such as SWT-bench~\citep{mundler2024swt} have successfully integrated post-generation selection heuristics based on CodeT to boost performance on repository-level tasks. 

The concept of an unobservable reward component ($r_{hid}$) is also highly relevant in agentic backprompting. When agentic systems are evaluated on benchmarks like SWE-bench \citep{jimenez2024swebench}, providing oracle information by applying test patches to the project (a "reproduction" setting) effectively reveals $r_{hid}$ to the agent. As predicted by our theory, eliminating the irreducible regret $\Delta$ can have a profound impact; notably, concurrent work by~\citet{li2026oracle} shows that under such oracle information the performance of GPT-5 on SWE-bench Live improves from 30\% to 75\%. \looseness -1

\begin{wrapfigure}{r}{0.5\textwidth}
    \centering
    \vspace{-15pt}\includegraphics[width=0.985\linewidth]{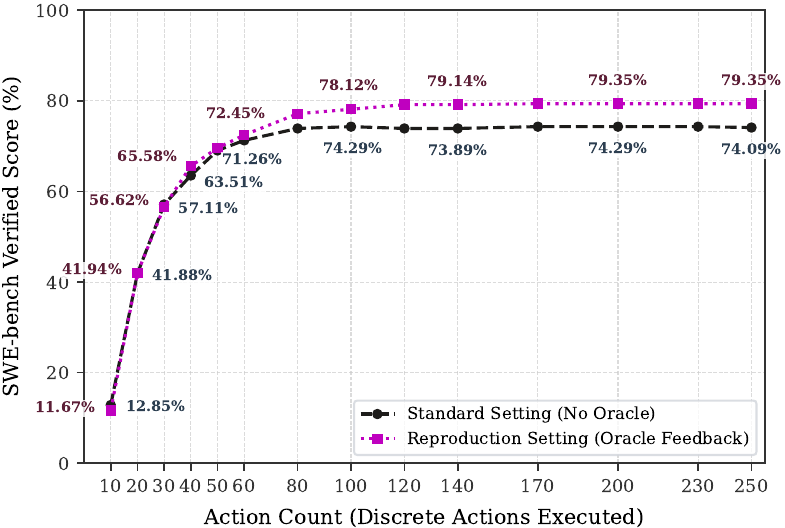}
    \caption{\begin{rebuttal}Solve rate (\%) of MiniMax M2.7 during 250 environment interactions on SWE-bench Verified. The reproduction setting applies oracle test patches, and thus matches $\Delta = 0$ in Corollary~\ref{cor:reward_bound}.\end{rebuttal}}
    \label{fig:swe-bench}
    \vspace{-15pt}
\end{wrapfigure}

\begin{rebuttal}
To verify this phenomenon empirically in a repository-level agentic setting, we conduct experiments on SWE-bench Verified~\citep{jimenez2024swebench} using MiniMax M2.7 (see Figure~\ref{fig:swe-bench}). In this setting, the sequential bandit formulation adapts to a slightly more complex structure: while it still fits our contextual Thompson sampling framework, the actions are no longer as neat as a single code-test generation pair. Instead, a sequence of multiple discrete subactions (such as viewing files, invoking bash commands, and editing lines) collectively constitutes a full edit of the codebase, which is only subsequently verified via execution against self-tests.

Despite this multi-subaction abstraction layer, our framework's predictive power remains robust. By measuring the performance gap at convergence between the Reproduction Setting (\textsc{Oracle-Test Feedback}) and the Standard Setting (\textsc{Self-Test Feedback}), we can directly estimate the irreducible regret $\Delta$, yielding $5.26\%$.

Interestingly, this positions SWE-bench Verified between QiskitHumanEvalSim ($\Delta=16.08\%$) and LeetCodeDataset ($\Delta\approx 0\%$). The relatively low irreducible regret implies that real-world issue descriptions often contain rich contextual clues about the target algorithm, explaining the strong empirical success of multi-round backprompting agents when navigating repository-level codebases.

\end{rebuttal}

\section{Conclusion}
In this work, we provide a probabilistic framework for LLM-based test-driven code generation. We theoretically formalize two dominant paradigms: post-generation selection as environment-aware pass@k maximization, and iterative backprompting as in-context Thompson sampling. Our analysis proves that ``soft" estimators strictly dominate ``hard" ones in terms of SNR. Moreover, we identify an ``irreducible regret" in the regret bound of Thompson sampling, where task ambiguity limits performance regardless of feedback quality. These theoretical insights are confirmed experimentally across \begin{rebuttal}four\end{rebuttal} diverse benchmarks using \begin{rebuttal}seven\end{rebuttal} different frontier open-weight large language models. Finally, a novel benchmark, \textsc{QiskitHumanEvalSimX}, is introduced to corroborate our findings. As LLMs are increasingly integrated into critical software development loops, understanding the theoretical mechanisms of their environment interactions becomes paramount. Our results suggest that progress lies not just in improving zero-shot model performance, but also in reducing the ambiguity of user intent and designing architectures that can efficiently process long execution histories.

\section*{Acknowledgment}
This work is supported by the Swiss National Science Foundation (SNSF), grant 10002666.

\bibliographystyle{plainnat}
\bibliography{references}

\newpage
\appendix

\setcounter{figure}{0}
\renewcommand{\thefigure}{S\arabic{figure}}
\setcounter{table}{0}
\renewcommand{\thetable}{S\arabic{table}}

\section{Limitations, negative results, and societal impact}
Our investigation isolates two core components of coding agents~\citep{yang2024swe}: post-generation selection and iterative revision via environment feedback. To ensure controlled experiments with well-defined environment interaction, we emphasize competitive coding tasks rather than multi-file editing workflows. We encourage future work to \begin{rebuttal}put more emphasis on\end{rebuttal} fully agentic settings. While our framework defines $\Delta$ and identifies it as a key limiting factor, future work could actively proxy $\Delta$ during inference to trigger user clarification, preventing wasted compute on unresolvable ambiguity.sc

\begin{rebuttal}
\paragraph{Connection to Classical Learnability and Negative Results} Our probabilistic framework assumes programs and test suites to lie within a space of token strings with maximum length $L$. Moreover, all our definitions, such as code similarity, assume strict execution timeouts when running programs / test suites. These assumptions are required to operate in a regime where the Halting Problem~\citep{turing1936computable} or the negative results of~\citeauthor{gold1967language}'s \citeyearpar{gold1967language} identifiability in the limit are ineffective. In contrast to \citeauthor{shapiro1982algorithmic}'s \citeyearpar{shapiro1982algorithmic} algorithmic debugging framework, we do not assume access to an oracle agent that can answer queries about the target program's exact behavior. Instead, we only rely on the initial task description, which is a finite-length text. As such, we cannot guarantee convergence to the true optimal program. Instead, our paper considers a strongly Bayesian perspective, where all statements are with respect to the learned LLM prior.

\paragraph{Modeling Assumptions and the Learned Prior} Our theoretical interpretation of iterative backprompting relies heavily on modeling the LLM as a Thompson sampler. However, LLM pre-training only promotes approximate Thompson sampling behavior, with only vacuous PAC-learning bounds~\citep{valiant1984theory} for LLM-sized hypothesis spaces. Consequently, formal statements that build upon the assumption of exact Thompson sampling (such as Theorem~\ref{thm:reward_bound}) should be interpreted as qualitative or stylized characterizations of LLM behavior rather than absolute mathematical guarantees. Additionally, the underlying reward function in software engineering is deterministic, not probabilistic. Our theoretical results are purely Bayesian; they reflect the expected reward with respect to the LLM’s learned prior belief, rather than an absolute ground truth.

\paragraph{Irreducible Regret as a Fundamental Limit} Even if an LLM were a perfect Thompson sampler, Theorem~\ref{thm:reward_bound} would still identify an irreducible regret ($\Delta$). This explicitly establishes a negative result: absolute convergence to the optimal, ground-truth solution cannot be guaranteed due to the inherent ambiguity of text-only, informal task descriptions, even under ideal assumptions.

\paragraph{Computational Constraints} Both of the investigated types of environment interaction strategies suffer from severe computational constraints. Post-generation selection heuristics necessitate executing every generated code against every generated test, resulting in a quadratic runtime execution bottleneck. Iterative backprompting avoids this runtime execution cost, but introduces a massive memory overhead, since appending execution histories linearly grows the context window.
\end{rebuttal}

\section*{Impact statement}
This paper advances the theoretical understanding of test-driven code generation via LLMs. As a theoretical work, its societal impact will be mostly indirect. More generally, a positive societal impact of research on LLM-assisted coding is improved developer productivity and software reliability. However, there are potential negative societal impacts to consider. Highly capable coding assistants may lead to shifts in the software engineering labor market, e.g., reducing entry-level job opportunities.

\newpage
\section{Related work}

\paragraph{LLMs for code generation} Since the introduction of Codex \citep{chen2021evaluating}, a model fine-tuned on GitHub, code generation has become a central application of LLMs. Modern general-purpose models now routinely include substantial code data during training \citep{yang2025qwen3technicalreport}.

\paragraph{Post-generation selection strategies} Selection heuristics have become a popular approach for improving performance without retraining~\citep{wang2023selfconsistency}. \textsc{AlphaCode} \citep{li2022competition} popularized selecting generations based on pass rates on ground-truth tests and behavioral similarity under self-generated inputs. In scenarios where external tests are unavailable, heuristics like \textsc{MBR-Exec} \citep{shi2022mbr} and \textsc{CodeT} \citep{chen2023codet} select solutions based on agreement groups---sets of codes that pass the same self-generated tests. Many heuristics, such as \textsc{CodeT}, rely on ``hard" agreement (counting identical outputs). Our work theoretically formalizes these as environment-aware estimators of program correctness and investigates ``soft" alternatives that use fuzzy functional similarity to improve accuracy.

\paragraph{Iterative refinement via backprompting} Beyond selection, environment feedback is used to iteratively refine code. Frameworks like \textsc{Reflexion} \citep{shinn2023reflexion} use verbal reinforcement learning, asking LLMs to reflect on textual feedback to guide repairs. Similarly, \textsc{SelfRepair} \citep{olausson2024is} executes a repair step based on external test failures. While these methods assume valid, discriminative tests provided, our framework addresses the more common setting where assistants rely on self-generated tests. In this setup, we frame backprompting as approximate Thompson sampling~\citep{thompson1933likelihood, russo2014learning}. This allows us to theoretically identify an irreducible regret: a term in the regret bound caused by the inherent task ambiguity.

\paragraph{Bayesian perspectives on in-context learning} Recent theoretical work by \citeauthor{xie2022implicit} \citeyearpar{xie2022implicit} and \citeauthor{chlon2025llmsbayesianexpectationrealization} \citeyearpar{chlon2025llmsbayesianexpectationrealization} demonstrates that in-context learning conducts approximate Bayesian inference. \citeauthor{de2025simplifying} \citeyearpar{de2025simplifying} argue that pre-training on sequence data with cross-entropy loss naturally yields models capable of Thompson sampling. While \citet{menet2026thompson, menet2026poetsuncertaintyawarellmoptimization} explore Thompson sampling via fine-tuning of LLMs, we study the emergent Thompson sampling capabilities of off-the-shelf LLMs via in-context conditioning. This allows us to develop and apply regret bounds from Bayesian optimization to explain the limitations of backprompting strategies.

\subsection{Post-generation selection strategies as estimators of code correctness}\label{sec:popular_heuristics}
Our framework allows us to identify several well-established selection heuristics as maximizing an estimate of the probability of correctness $p(\mathcal N_c^\infty)$, see the developments in Section~\ref{sec:post_generation_selection_via_feedback}. Table~\ref{tab:types_of_self_consistency} provides a summary of the code selection literature based on our formalism.
In the experimental Section~\ref{sec:experiments}, we focus on boolean test outputs $\mathbb O$, i.e., the canonical case of test-driven development.

\paragraph{\textsc{MBR-Exec}}
Under the minimum Bayes risk decoding strategy \citep{shi2022mbr}, a set of codes $c_1, \ldots, c_n \in \mathcal V^{\le L}$ and test suites $t_1, \ldots, t_m \in \mathcal V^{\le L}$ are sampled according to the LLM posteriors $p(c)$ and $p(t)$ given a task description $d \in \mathcal V^{\le L}$. Then,
a code $\hat c = \arg\min_{c_j \in c_1, \ldots, c_n} \sum_{i\not=j} \ell(c_i, c_j)$ is selected to minimize an empirical estimate of the Bayes risk $\mathbb E_{c^\prime \sim p}[\ell(c, c^\prime)]$ with risk function $\ell$. The risk $\ell$ is chosen either as $-\widehat{\smash{\mathrm{sim}}\vphantom{m}}_{p,e;m}^{\infty}(c, c^\prime)$ or $-\widehat{\smash{\mathrm{sim}}\vphantom{m}}_{p,e;m}^1(c, c^\prime)$, resulting in \textsc{MBR-Exec Hard} and \textsc{MBR-Exec Soft}, respectively.  

\paragraph{\textsc{AlphaCode} and \textsc{FunCoder}}
\citeauthor{li2022competition} \citeyearpar{li2022competition} and \citeauthor{chen2024divideandconquer} \citeyearpar{chen2024divideandconquer} do not sample boolean-valued test cases from the language model, but rather execution scripts that invoke the code at arbitrary inputs. In our formalism, they apply generalized test suites $t$ that return a vector of non-boolean outputs $\mathbb A^{|t|}$, which is then compared to estimate functional similarity. Whereas \textsc{AlphaCode} relies on functional equivalence by estimating $\widehat{\smash{\mathrm{sim}}\vphantom{m}}_{p,e;m}^{\infty}(c, c^\prime)$, \textsc{FunCoder} relaxes this to functional similarity $\widehat{\smash{\mathrm{sim}}\vphantom{m}}_{p,e;m}^1(c, c^\prime)$.

\paragraph{\textsc{MaxPass}} \citet{lahiri2023interactivecodegenerationtestdriven, roziere2022leveraging, le2022coderl, chen2024b4} consider Monte Carlo estimators of the dual expression of code correctness $\mathbb P_{t \sim p}[R(c, t\,\vert\, e) = 1]$, albeit only in the setting of single-test test-suites, where the hard and soft estimators coincide. Theoretically, if a model were perfectly code-test calibrated (Definition~\ref{def:code_test_calibration}), code-based probability of correctness (\textsc{MBR-Exec}) and test-based probability of correctness (\textsc{MaxPass}) would yield the exact same result in expectation. However, as we demonstrate in Section~\ref{sec:experiments_post_generation_selection_heuristics}, \textsc{MaxPass} consistently outperforms \textsc{MBR-Exec}, indicating that current LLMs are miscalibrated. Since Monte Carlo estimators of probability of maximality have high variance \citep{menet2025lite}, one may instead consider a soft measure of code correctness such as $\hat{\mathbb E}_{t_{1:m}}[R(c, t\,\vert\, e)]$.

\paragraph{\textsc{CodeT}}
\citeauthor{chen2023codet} \citeyearpar{chen2023codet} multiply the objective in the \textsc{MBR-Exec Hard} strategy with the number of passed test cases. According to our formalism, they take the geometric mean of two dual estimators of self-consistency: $\hat{p}_{n,m}(\mathcal N^\infty_c)$ and $\hat{\mathbb E}_{t_{1:m}}[R(c,t\vert e)]$. With their method, they argue for increased robustness. Based on the insights from Section~\ref{sec:similarity_vs_equivalence}, we extend \textsc{CodeT} with a soft variant (\textsc{CodeT Soft}) that adopts the smooth estimator $\hat{p}_{n,m}(\mathcal N_c^1)$. As demonstrated in the experimental section, our smooth variant of \textsc{CodeT} performs consistently better than the original hard variant.

\newpage
\section{Additional experiments}\label{sec:additional_experiments}

\subsection{Improvements of post-generation selection as a function of the number of rounds}\label{sec:post_selection_appendix}
We complement Table~\ref{tab:post_generation_selection_methods_results} with Figure~\ref{fig:post_generation_selection_heuristics_line_plot}, which details the accuracy improvements achieved after generation through the considered selection heuristics. We plot the improvement as a function of the number of rounds of code and test suite generations. In several settings, the accuracy improvements do not yet saturate after $10$ rounds, but the compute cost of post-generation selection starts to dominate the generation cost (using a Qwen3-235B-A22B-Instruct, even for massive parallelism with $64$ AMD EPYC 7763 CPU cores for selection and only a single node with $8$ NVIDIA A100 80GB GPUs for generation). See Section~\ref{sec:runtimes} for runtime measurements. This runtime bottleneck stems from the quadratic scaling of post-generation selection, as every code must be executed on every test suite.

\begin{figure}[h]
    \centering
    \includegraphics[width=\linewidth]{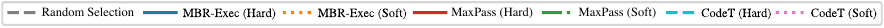}
    
    \begin{subfigure}[b]{0.333\textwidth}
        \centering
        \captionsetup{justification=centering}\includegraphics[width=\linewidth]{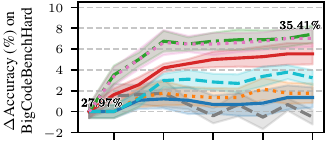}
    \end{subfigure}%
    \begin{subfigure}[b]{0.333\textwidth}
        \centering
        \captionsetup{justification=centering}\includegraphics[width=\linewidth]{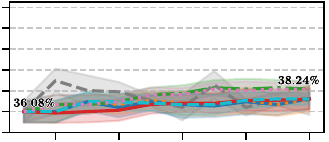}
    \end{subfigure}%
    \begin{subfigure}[b]{0.333\textwidth}
        \centering
        \captionsetup{justification=centering}\includegraphics[width=\linewidth]{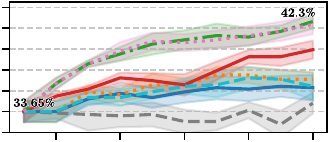}
    \end{subfigure}

\begin{subfigure}[b]{0.333\textwidth}
        \centering
        \captionsetup{justification=centering}\includegraphics[width=\linewidth]{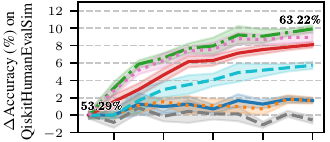}
    \end{subfigure}%
    \begin{subfigure}[b]{0.333\textwidth}
        \centering
        \captionsetup{justification=centering}\includegraphics[width=\linewidth]{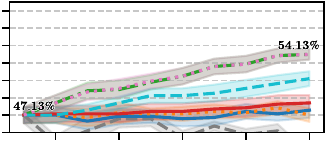}
    \end{subfigure}%
    \begin{subfigure}[b]{0.333\textwidth}
        \centering
        \captionsetup{justification=centering}\includegraphics[width=\linewidth]{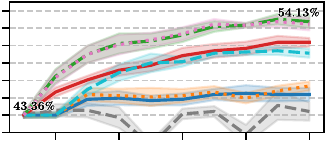}
    \end{subfigure}

    \begin{subfigure}[b]{0.333\textwidth}
        \centering
        \captionsetup{justification=centering}\includegraphics[width=\linewidth]{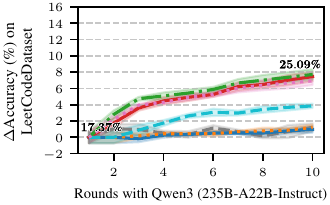}
    \end{subfigure}%
    \begin{subfigure}[b]{0.333\textwidth}
        \centering
        \captionsetup{justification=centering}\includegraphics[width=\linewidth]{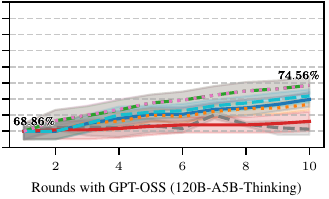}
    \end{subfigure}%
    \begin{subfigure}[b]{0.333\textwidth}
        \centering
        \captionsetup{justification=centering}\includegraphics[width=\linewidth]{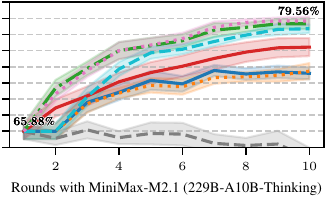}
    \end{subfigure}
    
    \caption{Pass@1 (\%) improvements of common post-generation selection heuristics (see Table~\ref{tab:types_of_self_consistency} and Table~\ref{tab:post_generation_selection_methods_results}) across BigCodeBenchHard, QiskitHumanEvalSim, and LeetCodeDataset. Qwen3-235B-A22B-Instruct-2507, GPT-OSS-120B, or MiniMax-M2.1 are used to independently generate code and tests across multiple rounds. Results are reported as mean ($\pm$ standard error) over 5 seeds.}
    \label{fig:post_generation_selection_heuristics_line_plot}
\end{figure}

\subsection{Approximate in-context Thompson sampling with GPT-OSS-120B and MiniMax-M2.1}\label{sec:backprompt_gpt_mm}
Figure~\ref{fig:in_context_thompson_sampling_line_plot_complete} reports additional results on GPT-OSS-120B and MiniMax-M2.1. We find that Qwen3-235B-A22B-Instruct is a more effective approximate Thompson sampler. Indeed, given \textsc{Oracle-Test Feedback} (purple dotted line), Qwen3-235B-A22B-Instruct yields significantly stronger accuracy improvements than MiniMax-M2.1, whose gains saturate after just one round. Moreover, both GPT-OSS-120B and MiniMax-M2.1 generate test suites that are much worse. In fact, given \textsc{Self-Test Feedback}, the \textsc{Self-Test Accuracy} (cyan dashed line) is much lower than the \textsc{Oracle-Test Accuracy} 
\begin{wrapfigure}{r}{0.4\linewidth}
    \centering
    \vspace{-7pt}
    \includegraphics[width=\linewidth]{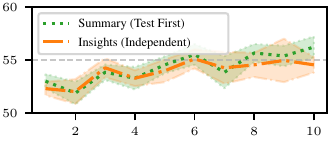}
    \caption{Pass@1 of in-context Thompson sampling via backprompting of self-test feedback using Qwen3-235B-A22B-Instruct-2507 on QiskitHumanEvalSim.}
    \label{fig:ind_vs_tfs_on_qiskit_human_eval_qwen3_235B}
    \vspace{-20pt}
\end{wrapfigure}
(black solid line). As a result, these two models receive highly noisy environment feedback and thus no accuracy gains over the zero-shot performance can be achieved.

Note that, whereas Figure~\ref{fig:in_context_thompson_sampling_line_plot} reports on \textit{summary concatenation} and \textit{test-first} factorization ($p(c,t) = p(c \vert t) \cdot p(t)$, see Section~\ref{sec:ablation}), Figure~\ref{fig:in_context_thompson_sampling_line_plot_complete} reports our initial experiments using \textit{insight reformulation} and \textit{independent} factorization ($p(c,t) = p(c) \cdot p(t)$). This avoids weeks of compute effort. Note that while \textit{summary concatenation} and \textit{test-first} factorization result in superior performance, the optimization trajectory behaves rather similarly, see also Figure~\ref{fig:ind_vs_tfs_on_qiskit_human_eval_qwen3_235B}.

\begin{figure}[h]
    \centering
    \includegraphics[width=\linewidth]{graphics/in_context_oracle_legend.pdf}
    \begin{subfigure}[b]{0.33\textwidth}
        \centering
        \captionsetup{justification=centering}\includegraphics[width=\linewidth]{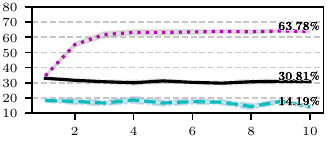}
        \caption{\footnotesize Qwen3 on BigCodeBenchHard}
    \end{subfigure}%
    \hfill
    \begin{subfigure}[b]{0.33\textwidth}
        \centering
        \captionsetup{justification=centering}\includegraphics[width=\linewidth]{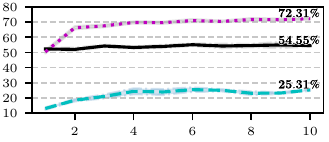}
        \caption{\footnotesize Qwen3 on QiskitHumanEvalSim}
    \end{subfigure}%
    \hfill
    \begin{subfigure}[b]{0.33\textwidth}
        \centering
        \captionsetup{justification=centering}\includegraphics[width=\linewidth]{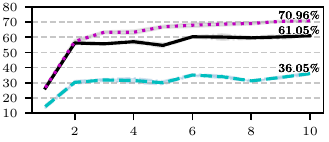}
        \caption{\footnotesize Qwen3 on LeetCodeDataset}
    \end{subfigure}

    \begin{subfigure}[b]{0.33\textwidth}
        \centering
        \captionsetup{justification=centering}\includegraphics[width=\linewidth]{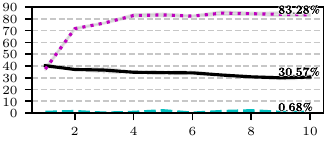}
        \caption{\footnotesize GPT-OSS on BigCodeBenchHard}
    \end{subfigure}%
    \hfill
    \begin{subfigure}[b]{0.33\textwidth}
        \centering
        \captionsetup{justification=centering}\includegraphics[width=\linewidth]{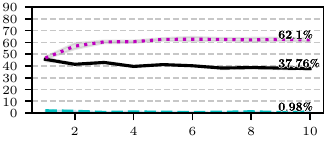}
        \caption{\footnotesize GPT-OSS on QiskitHumanEvalSim}
    \end{subfigure}%
    \hfill
    \begin{subfigure}[b]{0.33\textwidth}
        \centering
        \captionsetup{justification=centering}\includegraphics[width=\linewidth]{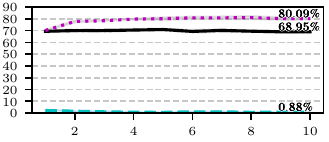}
        \caption{\footnotesize GPT OSS on LeetCodeDataset}
    \end{subfigure}

    \begin{subfigure}[b]{0.33\textwidth}
        \centering
        \captionsetup{justification=centering}\includegraphics[width=\linewidth]{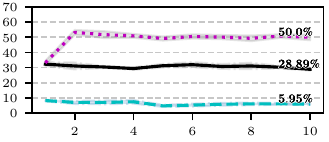}
        \caption{\footnotesize MiniMax on BigCodeBenchHard}
    \end{subfigure}%
    \hfill
    \begin{subfigure}[b]{0.33\textwidth}
        \centering
        \captionsetup{justification=centering}\includegraphics[width=\linewidth]{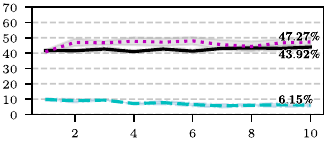}
        \caption{\footnotesize MiniMax on QiskitHumanEvalSim}
    \end{subfigure}%
    \hfill
    \begin{subfigure}[b]{0.33\textwidth}
        \centering
        \captionsetup{justification=centering}\includegraphics[width=\linewidth]{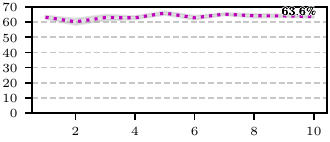}
        \caption{\footnotesize MiniMax on LeetCodeDataset}
    \end{subfigure}
    
    \caption{Pass@1 (\%) of approximate in-context Thompson sampling via backprompting with Qwen3-235B-A22B-Instruct-2507, GPT-OSS-120B, and MiniMax-M2.1 across BigCodeBenchHard, QiskitHumanEvalSim, and LeetCodeDataset. In the language of Theorem~\ref{thm:reward_bound}, the black solid line corresponds to the true reward $r$, and the cyan dashed line to $r_{obs}$. The purple dotted line corresponds to the case of known $r_{hid}$, which according to Corollary~\ref{cor:reward_bound} permits global convergence to $x^*$. These experiments were conducted with \textit{insight reformulation} and \textit{independent} factorization ($p(c,t) = p(c) \cdot p(t)$). Results are reported as mean ($\pm$ standard error) over 5 random seeds.}
    \label{fig:in_context_thompson_sampling_line_plot_complete}
\end{figure}

\subsection{Backprompting with hybrid Qwen3-Next, Qwen3-Coder-Next, and \begin{rebuttal}Qwen3.5-397B-A17B\end{rebuttal}}\label{sec:hybrid} 
\begin{wraptable}{r}{0.77\linewidth}
\vspace{-10pt}
    \begin{tabular}{cccc}
         &  BigCode & Qiskit & LeetCode\\
         \hline
         Qwen3-235B-A22B & 33.38 -> 33.11 & 53.0\phantom{0} -> 56.22 & 24.65 -> 72.11\\
         Qwen3-Next & 35.95 -> 32.7\phantom{0} & 50.35 -> 56.08 & 28.07 -> 67.28\\ 
         Qwen3-Coder-Next & 36.08 -> 30.27 & 43.5\phantom{0} -> 52.31 & 39.3\phantom{0} -> 64.39\\
         \begin{rebuttal}Qwen3.5-397B-A17B\end{rebuttal} & \begin{rebuttal}37.97 -> 37.3\phantom{0}\end{rebuttal} & \begin{rebuttal}55.66 -> 59.3\phantom{0}\end{rebuttal} & \begin{rebuttal}66.96 -> 73.98\end{rebuttal}
    \end{tabular}
    \caption{Initial and final pass@1 (\%) of members of the Qwen3(.5) model family during 10 rounds of backprompting self-test feedback. We report the mean across 5 seeds. For detailed trajectories, see Figure~\ref{fig:backprompting_qwen_family}.}
    \label{tab:backprompting_qwen_family}
\end{wraptable}
In addition, we evaluate Qwen3-Next, Qwen3-Coder-Next, and \begin{rebuttal}Qwen3.5-397B-A17B\end{rebuttal}~\citep{qwen3codernexttechnicalreport}. \begin{rebuttal}The Next variants\end{rebuttal} are 80B parameter models of which only 3B are active, and \begin{rebuttal}all three\end{rebuttal} are built on a hybrid DeltaNet-Transformer architecture~\citep{vaswani2017, yang2024parallelizing}. 
\begin{figure}[b]
    \centering
    \includegraphics[width=\linewidth]{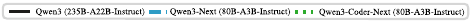}
    \begin{subfigure}[b]{0.33\textwidth}
        \centering
        \captionsetup{justification=centering}\includegraphics[width=\linewidth]{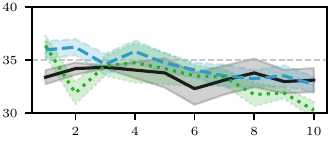}
        \caption{BigCodeBenchHard}
    \end{subfigure}
    \hfill
    \begin{subfigure}[b]{0.32\textwidth}
        \centering
        \captionsetup{justification=centering}\includegraphics[width=\linewidth]{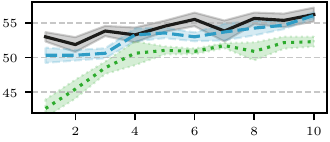}
        \caption{QiskitHumanEvalSim}
    \end{subfigure}
    \hfill
    \begin{subfigure}[b]{0.333\textwidth}
        \centering
        \captionsetup{justification=centering}\includegraphics[width=\linewidth]{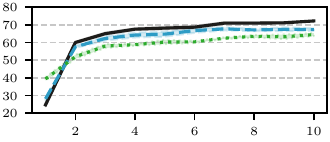}
        \caption{LeetCodeDataset}
    \end{subfigure}
    
    \caption{Pass@1 (\%) of backprompting concatenated summaries of self-test feedback using the Qwen3 model family. Results are reported as mean ($\pm$ standard error) over 5 seeds.}
    \label{fig:backprompting_qwen_family}
\end{figure}
As shown in Table~\ref{tab:backprompting_qwen_family} and Figure~\ref{fig:backprompting_qwen_family}, these models achieve results competitive with the Transformer-based Qwen3-235B-A22B-Instruct under backprompting with \textsc{Self-Test Feedback}. Moreover, hybrid models integrate well with \textit{summary-concatenation} history compression, since the cost of the state-space-model part does not grow with the rounds. Overall, these results indicate that model size and architectural design are likely not what hold GPT-OSS and MiniMax-M2.1 back.

\subsection{Ablating model factorization and feedback compression \begin{rebuttal}/ compactification\end{rebuttal}}\label{sec:ablation}
In Table~\ref{tab:in_context_thompson_sampling_feedback_ablation_results}, we investigate two degrees of freedom in the implementation of in-context Thompson sampling. To start, we consider different factorizations of $p(c,t)$: \textit{independent}, \textit{code first}, and \textit{test first}. First generating the test suite, and then a matching implementation conditioned on the suite works best. Next, we consider different LLM-based compressors of test reports $U(c,t\,\vert\, e)$. Compression is necessary, since concatenation of the test reports themselves quickly exceeds the available memory. We ablate two compressors: iterative reformulation of past insights in light of a new test report, or concatenation of summaries of each test report. A detailed description of each is provided in Appendix~\ref{sec:compressor_description}. As our results demonstrate, \textit{summary concatenation} outperforms \textit{insight reformulation}, although, at the cost of additional memory and compute. Unless otherwise mentioned, we always report on \textit{summary concatenation} with the factorization $p(c\mid t)\cdot p(t)$.

\begin{table*}[h]
    \centering
    \caption{Pass@1 (\%) after 10 rounds of in-context Thompson sampling (with improvements during generation in parenthesis) for different factorization of $p(c,t)$ and different test report compression methods. Results reported as mean ($\pm$ standard error) over 5 random seeds using Qwen3-235B-A22B-Instruct. For each column, the best configuration is emphasized and the second best is underlined.\looseness -1}
    \label{tab:in_context_thompson_sampling_feedback_ablation_results}
    \resizebox{\linewidth}{!}{
    \begin{tabular}{ccccc}
    \toprule
        Compression & Factorization & BigCodeBenchHard & QiskitHumanEvalSim & LeetCodeDataset\\
        \cmidrule(r){0-1}\cmidrule(lr){3-3}\cmidrule(lr){4-4}\cmidrule(lr){5-5}
Insights Reformulation & $p(c) \cdot p(t)$ & $30.81^{\pm 0.28}$ ($-2.16^{\pm 0.32}$) & $54.55^{\pm 0.33}$ ($2.24^{\pm 0.33}$) & $61.05^{\pm 0.32}$ ($34.56^{\pm 0.48}$) \\
Insights Reformulation & $p(t\mid c) \cdot p(c)$ & $31.89^{\pm 0.45}$ ($-0.68^{\pm 0.51}$) & $54.27^{\pm 0.39}$ ($2.94^{\pm 0.90}$) & $58.60^{\pm 0.47}$ ($34.65^{\pm 0.72}$) \\
Insights Reformulation & $p(c\mid t) \cdot p(t)$ & $\textbf{34.05}^{\pm 0.31}$ ($-0.81^{\pm 0.46}$) & $54.97^{\pm 0.19}$ ($3.08^{\pm 0.40}$) & $60.61^{\pm 0.48}$ ($36.40^{\pm 0.30}$) \\
\cmidrule(r){0-1}\cmidrule(lr){3-3}\cmidrule(lr){4-4}\cmidrule(lr){5-5}
Summary Concatenation & $p(c) \cdot p(t)$ & $32.84^{\pm 0.45}$ ($-0.27^{\pm 0.49}$) & $55.52^{\pm 0.36}$ ($3.92^{\pm 0.46}$) & $\underline{71.14}^{\pm 0.24}$ ($42.02^{\pm 0.32}$) \\
Summary Concatenation & $p(t\mid c) \cdot p(c)$ & $32.03^{\pm 0.44}$ ($-0.41^{\pm 0.58}$) & $\underline{56.08}^{\pm 0.55}$ ($5.45^{\pm 0.92}$) & $68.42^{\pm 0.31}$ ($44.74^{\pm 1.00}$) \\
Summary Concatenation & $p(c\mid t) \cdot p(t)$ & $\underline{33.11}^{\pm 0.51}$ ($-0.27^{\pm 0.75}$) & $\textbf{56.22}^{\pm 0.45}$ ($3.22^{\pm 0.70}$) & $\textbf{72.11}^{\pm 0.08}$ ($47.46^{\pm 0.39}$) \\
        \bottomrule
    \end{tabular}
    }
\end{table*}

\begin{wraptable}{r}{0.55\textwidth}
    \centering
    \vspace{-13.3pt}
    \caption{Pass@1 (\%) comparison of backprompting and post-gen selection. Each solution produced 10 code-test pairs. We use a Qwen3-235B-A22B-Instruct-2507 and report mean ($\pm$ standard error) over 5 seeds.}
    \label{tab:post_generation_selection_methods_with_in_context_thompson_sampling_results}
    \small
    \resizebox{0.95\linewidth}{!}{
    \setlength{\tabcolsep}{2pt}
    \begin{tabular}{llccc}
 Selection & Backprompting  & BigCode & Qiskit & LeetCode\\
        \cmidrule(r){1-2} \cmidrule(lr){3-3}\cmidrule(lr){4-4}\cmidrule(lr){5-5}
\textsc{Random} & No & $27.43^{\scriptscriptstyle\pm 0.31}$ & $52.73^{\scriptscriptstyle\pm 0.19}$ & $18.60^{\scriptscriptstyle\pm 0.22}$ \\[0.3ex] 
\textsc{Last} & \textsc{Self-Test FB} & $\mathbf{33.11}^{\scriptscriptstyle\pm 0.51}$ & $\mathbf{56.22}^{\scriptscriptstyle\pm 0.45}$ & $\mathbf{72.11}^{\scriptscriptstyle\pm 0.08}$ \\
\cmidrule(r){1-2} \cmidrule(lr){3-3}\cmidrule(lr){4-4}\cmidrule(lr){5-5}
\textsc{MP Soft} & No & $35.41^{\scriptscriptstyle\pm 0.38}$ & $\underline{\mathbf{63.22}}^{\scriptscriptstyle\pm 0.19}$ & $25.09^{\scriptscriptstyle\pm 0.29}$ \\[0.3ex] 
\textsc{MP Soft} & \textsc{Self-Test FB} & $\underline{\mathbf{36.22}}^{\scriptscriptstyle\pm 0.47}$ & $60.84^{\scriptscriptstyle\pm 0.36}$ & $\underline{\mathbf{73.42}}^{\scriptscriptstyle\pm 0.17}$ \\
\bottomrule
    \end{tabular}}
    \vspace{-20pt}
\end{wraptable}

\subsection{Combining backprompting with post-generation selection}
Consider Table~\ref{tab:post_generation_selection_methods_with_in_context_thompson_sampling_results}, which reports on the \textsc{MaxPass} selection heuristic applied on top of in-context Thompson sampling. The $10$ rounds of code and test suite generations are integrated into post-generation selection as if independently generated. The pass@1 of Thompson sampling still benefits from post-generation selection, but the improvements are less pronounced than with independent generation. \looseness -1

\subsection{Computational cost of environment interaction and in-context Thompson sampling}\label{sec:runtimes}

While environment-aware selection and in-context Thompson sampling yield large accuracy improvements, they add substantial compute.

\textbf{Post-generation selection:} Heuristics like \textsc{MBR-Exec}, \textsc{MaxPass}, and \textsc{CodeT} require cross-executing $n$ generated implementations against $m$ generated test suites, resulting in a quadratic scaling of runtime execution cost, $\mathcal{O}(n \times m)$. In practice, this execution bottleneck becomes quickly apparent: e.g., on QiskitHumanEvalSimX with Qwen3-235B-A22B-Instruct-2507, using a single node with an AMD EPYC 7763 64-Core CPU and 8 NVIDIA A100 (80GB) GPUs, running all implementations on all test suites across the $143$ dataset entries for $10$ rounds with $5$ seeds takes $21$ hours and $32$ minutes, dominating the $12$ hours and $42$ minutes required for LLM based code-test generation. \begin{rebuttal}This motivates using soft heuristics with lower signal-to-noise ratio (see Theorem~\ref{thm:snr_dominance}), since they require fewer code/test generations to provide a meaningful selection signal.\end{rebuttal}

\textbf{Iterative backprompting:} Conversely, in-context Thompson sampling avoids the $\mathcal{O}(n \times m)$ runtime bottleneck by only evaluating the current code and test suite, but introduces additional demand on the LLM-based code-test generation. Given the more performant \textit{summary concatenation} compressor, see Appendix~\ref{sec:ablation}, extending the prompt with historical feedback linearly increases its length with every round. Given the quadratic complexity of attention in standard Transformers, prolonged environment interaction quickly hits a compute and memory ceiling. For instance, across our 10-round Qwen3-235B-A22B-Instruct-2507 experiments, the model processed 119 million total tokens, of which only 17 million were newly generated. This context overhead strongly suggests that model architectures with linear complexity, such as State-Space Models (SSMs) or Linear Transformers, will be critical for scaling iterative coding agents capable of dynamically retaining long execution histories.

\begin{rebuttal}
\subsection{Quality and discriminative power of generated tests}\label{sec:quality_and_discriminative_power_of_generated_tests}
To assess the composition and discriminative power of the self-generated test suites, we empirically analyzed all tests generated by our primary model (Qwen3-235B-A22B-Instruct) over the 10 rounds of backprompting across all datasets. Out of $261,566$ individual test cases, only $0.06\%$ ($166$ tests) contain a trivial \texttt{assert True}. An additional $7.20\%$ ($18,828$ tests) contain no explicit assert statements (relying instead on execution crashes to fail). The vast majority of tests contain programmatic assertions, demonstrating that the model assigns low probability mass to degenerate tests.

We further analyzed the distribution of the average test failure rate across all test cases (with the average computed across all generated programs). The overall mean failure probability is $36.83\%$. As depicted in Figure~\ref{fig:failure_prob_histogram}, the resulting distribution is notably bimodal. 
\end{rebuttal}

\begin{figure}[h]
    \centering
    \begin{tikzpicture}
        \begin{axis}[
            width=0.8\linewidth,
            height=6cm,
            ybar,
            bar width=20pt,
            ymin=0, ymax=35,
            xtick={1,2,3,4,5,6,7,8,9,10},
            xticklabels={
                {[0.0, 0.1)}, {[0.1, 0.2)}, {[0.2, 0.3)}, {[0.3, 0.4)}, {[0.4, 0.5)},
                {[0.5, 0.6)}, {[0.6, 0.7)}, {[0.7, 0.8)}, {[0.8, 0.9)}, {[0.9, 1.0)}
            },
            x tick label style={rotate=45, anchor=east, font=\scriptsize},
            ylabel={Percentage of Tests (\%)},
            xlabel={Failure Probability},
            nodes near coords={\pgfmathprintnumber\pgfplotspointmeta\%},
            nodes near coords align={vertical},
            every node near coord/.append style={font=\tiny},
            enlarge x limits=0.08,
            ymajorgrids=true,
            axis lines*=left
        ]
        \addplot[fill=cyan!60!black, draw=cyan!80!black] coordinates {
            (1, 29.74)
            (2, 27.28)
            (3, 8.57)
            (4, 0.00)
            (5, 1.45)
            (6, 2.01)
            (7, 0.74)
            (8, 0.00)
            (9, 0.73)
            (10, 29.49)
        };
        \end{axis}
    \end{tikzpicture}
    \vspace{-5pt}
    \caption{\begin{rebuttal}Distribution of average test failure rate (computed across all generated programs) for $261,566$ test cases generated by Qwen3-235B-A22B-Instruct on BigCode, Qiskit, and LeetCode.\end{rebuttal}}
    \label{fig:failure_prob_histogram}
\end{figure}
\begin{rebuttal}
This bimodal distribution indicates that the generated tests cover a broad spectrum of difficulty. Rather than providing a uniform evaluation signal, the tests tend to partition into two primary groups: permissive tests that the majority of generated programs pass (serving as basic sanity checks), and stringent tests that the majority of programs fail (e.g., evaluating edge cases).
\end{rebuttal}

\subsection{Empirical signatures of in-context Thompson sampling}\label{sec:empirical_signatures_ts}

\begin{rebuttal}
To justify modeling the LLM as an approximate Thompson sampler, we provide indirect empirical evidence of its sampling distribution. First, if the LLM ``lazily debugged'' the code and test suite to match, we would expect the self-generated tests to drift and become less faithful, e.g., by deleting difficult-to-satisfy test cases. However, as shown in Table~\ref{tab:test_suite_stability}, the test suite length and precision (positive predictive value) remain remarkably stable across $10$ rounds of backprompting.

\begin{table}[h]
    \centering
    \caption{Average test suite length and precision (positive predictive value) across 10 rounds of in-context Thompson sampling using the Qwen3-235B-A22B-Instruct model.}
    \label{tab:test_suite_stability}
    \small
    \begin{tabular}{lccc}
        \toprule
        Dataset & Round & Average Test Suite Length & Precision \\
        \midrule
        BigCodeBenchHard & 1 & 9.86 & 41.4\% \\
        BigCodeBenchHard & 5 & 10.50 & 41.9\% \\
        BigCodeBenchHard & 10 & 10.73 & 37.8\% \\
        \midrule
        QiskitHumanEvalSim & 1 & 7.17 & 73.1\% \\
        QiskitHumanEvalSim & 5 & 7.69 & 74.2\% \\
        QiskitHumanEvalSim & 10 & 7.80 & 71.2\% \\
        \midrule
        LeetCodeDataset & 1 & 11.59 & 74.1\% \\
        LeetCodeDataset & 5 & 11.33 & 78.9\% \\
        LeetCodeDataset & 10 & 11.45 & 78.6\% \\
        \bottomrule
    \end{tabular}
\end{table}

Next, if the model genuinely executed approximate Thompson sampling, its posterior belief would contract as environment feedback reduces uncertainty, yet remain stochastic due to the irreducible regret. In contrast, simple iterative debugging would be expected to collapse to a single deterministic solution. Table~\ref{tab:collision_probability} tracks the collision probability of consecutively generated codes (the probability they belong to the same functional equivalence class, estimating $\mathbb{E}_{c_{obs,1}^\ast, c_{obs,2}^\ast}[\mathrm{sim}_{p,e}^\infty(c_{obs,1}^\ast, c_{obs,2}^\ast)]$). The model's belief contracts significantly over time, yet always remains stochastic.

\begin{table}[h]
    \centering
    \caption{Collision probability per round of in-context Thompson sampling using Qwen3-235B-A22B-Instruct (averaged across a sliding window of size 3 to account for non-stationarity).}
    \label{tab:collision_probability}
    \small
    \begin{tabular}{cccc}
        \toprule
        Round & BigCodeBenchHard & QiskitHumanEvalSim & LeetCodeDataset\\
        \midrule
        4 & 56.40\% & 77.48\% & 75.29\% \\
        5 & 62.12\% & 80.70\% & 81.37\% \\
        6 & 66.08\% & 82.84\% & 83.13\% \\
        7 & 68.29\% & 84.38\% & 85.56\% \\
        8 & 71.71\% & 87.65\% & 87.02\% \\
        9 & 73.65\% & 87.69\% & 88.36\% \\
        10 & 72.93\% & 88.34\% & 88.89\% \\
        \bottomrule
    \end{tabular}
\end{table}
\end{rebuttal}

\subsection{Rank correlation and signal-to-noise ratio of selection heuristics}\label{sec:rank_correlation_snr}

\begin{rebuttal}
To explicitly connect the signal-to-noise ratio (SNR) dominance established in Theorem~\ref{thm:snr_dominance} to the empirical pass@1 gains, we analyze estimator variance and rank correlation with the oracle correctness. As Table~\ref{tab:snr_rank_correlation} demonstrates, the hard estimators (\textsc{MaxPass Hard}) actually match or slightly outperform the soft estimators (\textsc{MaxPass Soft}) in terms of rank correlation (Spearman's $\rho$ and Kendall's $\tau$). However, exactly as predicted by our theory, the soft estimators exhibit a vastly superior empirical signal-to-noise ratio. This confirms that the improved pass@1 of soft estimators is genuinely driven by their smoothing bias overcoming the high-variance bottleneck.

\begin{table}[h]
    \centering
    \caption{Rank correlation and SNR analysis at 5 and 10 rounds of independent generation using Qwen3-235B-A22B-Instruct. This is the primary setting investigated for post-generation selection.}
    \label{tab:snr_rank_correlation}
    \small
    \begin{tabular}{llcccc}
        \toprule
        Dataset & Strategy & \# Test Suites & Spearman's $\rho$ & Kendall's $\tau$ & SNR (bootstrapped) \\
        \midrule
        BigCodeBenchHard & \textsc{MaxPass Hard} & 5 & \textbf{0.500} & \textbf{0.494} & 3.01 \\
        & \textsc{MaxPass Soft} & 5 & 0.495 & 0.481 & \textbf{151.40} \\
        \addlinespace
        LeetCodeDataset & \textsc{MaxPass Hard} & 5 & \textbf{0.864} & \textbf{0.861} & 5.11 \\
        & \textsc{MaxPass Soft} & 5 & 0.804 & 0.793 & \textbf{553.83} \\
        \addlinespace
        QiskitHumanEvalSim & \textsc{MaxPass Hard} & 5 & \textbf{0.787} & \textbf{0.783} & 2.79 \\
        & \textsc{MaxPass Soft} & 5 & 0.692 & 0.680 & \textbf{53.53} \\
        \midrule
        BigCodeBenchHard & \textsc{MaxPass Hard} & 10 & \textbf{0.445} & \textbf{0.438} & 4.69 \\
        & \textsc{MaxPass Soft} & 10 & 0.420 & 0.399 & \textbf{207.08} \\
        \addlinespace
        LeetCodeDataset & \textsc{MaxPass Hard} & 10 & \textbf{0.815} & \textbf{0.811} & 7.12 \\
        & \textsc{MaxPass Soft} & 10 & 0.757 & 0.741 & \textbf{493.67} \\
        \addlinespace
        QiskitHumanEvalSim & \textsc{MaxPass Hard} & 10 & \textbf{0.731} & \textbf{0.726} & 3.96 \\
        & \textsc{MaxPass Soft} & 10 & 0.703 & 0.686 & \textbf{73.35} \\        
        \bottomrule
    \end{tabular}
\end{table}
\end{rebuttal}

\subsection{Comparison with B4 baseline}\label{sec:b4_comparison}

\begin{rebuttal}
To provide a comparison with recent Bayesian methods for post-generation selection, we benchmarked B4 \citep{chen2024b4}, which computes the MAP probability of an entire consensus set being the true correct solution. Note that in B4 a consensus set is defined via a Monte Carlo estimate of $\mathcal N_c^\infty$, tying into our theory developed in Section~\ref{sec:post_generation_selection_via_feedback}. Table~\ref{tab:b4_comparison} shows that \textsc{MaxPass Soft} consistently matches or outperforms B4, further illustrating the strength of our considered smooth estimators.

\begin{table}[h]
    \centering
    \caption{Pass@1 (\%) performance comparing B4 against \textsc{MaxPass Soft} as post-generation selection mechanisms after 10 rounds of independent generation using Qwen3-235B-A22B-Instruct.}
    \label{tab:b4_comparison}
    \small
    \begin{tabular}{lcc}
        \toprule
        Dataset & \textsc{B4} & \textsc{MaxPass Soft} \\
        \midrule
        BigCodeBenchHard & 33.24\% $\pm$ 0.48\% & \textbf{35.41\%} $\pm$ 0.75\% \\
        LeetCodeDataset & \textbf{25.18\%} $\pm$ 0.42\% & 25.09\% $\pm$ 0.57\% \\
        QiskitHumanEvalSim & 60.56\% $\pm$ 0.32\% & \textbf{63.22\%} $\pm$ 0.38\% \\
        \bottomrule
    \end{tabular}
\end{table}
\end{rebuttal}

\newpage
\section{Experimental details}\label{sec:experimental_details}
To ensure reproducibility, we complement Section~\ref{sec:experiments} with a more detailed description of the experimental setup. 
All experiments were conducted on compute nodes with an AMD EPYC 7763 64-Core CPU, 2TB RAM, and eight NVIDIA A100-SXM4 GPUs (80GB), running on Red Hat Enterprise Linux 9.4 with CUDA 12.4. Our code includes configuration files for each experiment and environment files specifying the necessary Python packages and their versions, and is publicly available at \url{https://github.com/IBM/theory-of-coding-agents} under a permissive license.

\subsection{Datasets}\label{sec:datasets}

\subsubsection{BigCodeBenchHard}
The partition of the 148/1140 most challenging tasks in BigCodeBench~\citep{zhuo2024bigcodebench}. BigCodeBench, released under the Apache 2.0 License, is designed for function-level code generation, adopting the format of HumanEval~\citep{chen2021evaluating}. The tasks in BigCodeBench require multiple external libraries across 7 domains including networking, data analysis, and visualization. The benchmark tests compositional reasoning and the ability to implement complex instructions that span 139 libraries with an average of 2.8 libraries per task.

\subsubsection{QiskitHumanEvalSim}
A novel subset of QiskitHumanEval~\citep{vishwakarma2024qiskit}, that we introduce. It consists of the subset of the 143/151 questions (excluding task ids 43, 97, 98, 122, 129, 133, 134, 146 of which 6 are classified as basic, 2 as intermediate, and none as difficult) that can be solved without access to a real-world quantum computer (i.e., simulated), democratizing the dataset. Due to the small fraction of training data being dedicated to Qiskit during typical LLM training, and since Qiskit's API has undergone multiple backward-incompatible changes, QiskitHumanEvalSim is a great testbed for runtime environment interaction. Note that QiskitHumanEval adopts the Apache 2.0 License.

\subsubsection{LeetCodeDataset v0.3.1}
228 diverse puzzle-style questions from LeetCode released after 2024-07-01 and collected by \citet{xia2025leetcodedatasettemporaldatasetrobust}. LeetCode typically includes example input-output pairs in natural language, significantly reducing the \text{irreducible regret} $\Delta$ identified in Corollary~\ref{cor:reward_bound} and, as we show, enabling more effective in-context Thompson sampling. The dataset was released under the MIT license.

\subsubsection{\textsc{QiskitHumanEvalSimX}}\label{sec:creating_qiskit_human_eval_simX}
\textsc{QiskitHumanEvalSimX} is a novel benchmark that we introduce in order to corroborate our theory, in particular the insight that effective in-context Thompson sampling requires a task description $d \in \mathcal V^{\le L}$ with small uncertainty on the latent algorithm $a \in \mathcal A$ (see Theorem~\ref{thm:reward_bound} and its Corollary~\ref{cor:reward_bound}). We improve the specification of the algorithm in $d$ by describing its target behavior in more detail, including input-output examples. More precisely, to create \textsc{QiskitHumanEvalSimX} from QiskitHumanEvalSim, we query Google Gemini 3 Pro (January 10, 2026) using the following prompt:\looseness -1

\begin{minipage}{\linewidth}
\begin{lstlisting}
Consider the following dataset entry:
```
{"task_id": "qiskitHumanEval/35", "prompt": "from qiskit.circuit.library [...]
```
Improve the docstring in the prompt by appending a natural language description of the test without leaking the canonical solution. The natural language description must be detailed, but must not contain code. If the test contains an input-output example, provide it. Don't explicitly mention the presence of tests. Don't change anything else. Return the improved dataset entry in the jsonl format.
\end{lstlisting}
\end{minipage}

In case the instructions were not followed or invalid symbols were returned, the query was repeated.
The resulting dataset was verified by running all canonical solutions against the oracle test suite. \begin{rebuttal}Moreover, to ensure that the prompt constraints were followed and no oracle-test leakage occurred, we manually inspected the generated dataset entries to verify that the canonical solutions and oracle tests were not directly revealed.\end{rebuttal}

\subsubsection{\begin{rebuttal}SWE-bench Verified\end{rebuttal}}
\begin{rebuttal}A human-validated subset of the SWE-bench evaluation framework~\citep{jimenez2024swebench}, consisting of real-world software engineering issues collected from open-source GitHub repositories. The benchmark transitions our evaluation from isolated function-level synthesis to a fully agentic coding paradigm, requiring language model agents to autonomously navigate complex multi-file codebases, execute bash tools, and implement code modifications. The dataset's irreducible regret $\Delta$ (see Corollary~\ref{cor:reward_bound}) is quite small, making it an excellent testing ground for multi-subaction contextual Thompson sampling. The dataset is released under the MIT license.\end{rebuttal}

\subsection{Models}\label{sec:extensive_model_description}
We consider \begin{rebuttal}five\end{rebuttal} frontier open-weight LLMs and two smaller hybrid LLMs, with an emphasis on qualitative differences such as reasoning vs instruct and general-purpose vs coding. We sample directly from the model with temperature $1.0$ and top-p set to $0.95$. After exceeding $32\ 768$ output tokens, generation is stopped. In case invalid output is generated (e.g., if no code or observations can be extracted) we repeat up to $9$ more times the same query to the LLM. The considered models are as follows: \looseness -1

\paragraph{Qwen3-235B-A22B-Instruct-2507} is the largest publicly available instruction-tuned model in the general-purpose Qwen3 model family~\citep{yang2025qwen3technicalreport}. Adopting the mixture of experts (MoE) architecture~\citep{shazeer2017}, only 22 out of the 235 billion parameters are active at inference. The model is \textbf{not trained for chain-of-thought}~\citep{wei2022chain}. It was released under Apache 2.0.

\paragraph{GPT-OSS-120B} is the largest publicly available model released by OpenAI \citep{agarwal2025gpt}. It is a general-purpose reasoning model trained for long \textbf{chain-of-thoughts}. Tallying a total of 120 billion parameters, but with only 5B active, it is optimized for high throughput. It follows Apache 2.0.

\paragraph{MiniMax-M2.1 and \begin{rebuttal}MiniMax-M2.7\end{rebuttal}} are state-of-the-art \textbf{coding} models released on the 23rd of December 2025 \begin{rebuttal}and 18th of March 2026\end{rebuttal}, building on the MiniMax-M1 architecture~\citep{chen2025minimax}. They are MoE with 229 billion total parameters, of which only 10B are active. They are trained for and \textbf{utilize chain-of-thought} for reasoning. The models are released under a modified MIT license.

\paragraph{Qwen3-Next and Qwen3-Coder-Next} are MoE models introduced by \citet{qwen3codernexttechnicalreport}. Of a total of 80B parameters only 3B are active. Building on a hybrid DeltaNet-Transformer architecture~\citep{vaswani2017, yang2024parallelizing}, they are particularly efficient when handling long contexts. We consider the instruct variants \textbf{without chain-of-thought}. They follow Apache 2.0.

\paragraph{\begin{rebuttal}Qwen3.5-397B-A17B\end{rebuttal}} \begin{rebuttal}is the largest hybrid Transformer-SSM model in our evaluation suite, introduced by Qwen Team in 2026. Tallying a total of 397 billion parameters with 17B active at inference, it expands on the hybrid architecture to efficiently manage large historical contexts. Crucially, and in contrast to its smaller siblings considered here, the model is fully trained for and \textbf{utilizes chain-of-thought} for complex algorithmic reasoning. It is released under the Apache 2.0 license.\end{rebuttal}

\subsection{Implementation and test suite generation prompts}
In the following, we list the system prompts used for implementation and test suite generation, adopting the factorization $p(c,t) = p(c \mid t) p(t)$. The prompts for the alternative factorizations are adjusted accordingly and listed in the code base.

\textbf{Test suite generation}

\begin{lstlisting}
You are an expert Python programmer.
Your task is to write pytests for an implementation that obeys the provided task description. You are provided observations from previous rounds of comparing potential implementations against potential test suites.

**1. Analyze and Plan**
 - Carefully reflect on the provided observations.
 - Your new test suite must address these observations while still covering the program's main functionality and common edge cases.

**2. Write the test suite**
 - Include all necessary imports for external libraries. The candidate implementation is available without imports.
 - Each pytest must be represented by a separate function starting with `test_`. In particular, do not define test classes.
 - Add a brief docstring to each pytest function explaining its purpose.
 - Do not number the pytests.
 - Provide all of the pytests in a single Python code block marked with triple backticks (```).
\end{lstlisting}

\textbf{Implementation generation}

\begin{lstlisting}
You are an expert Python programmer.
Your task is to write an implementation that obeys the provided task description. You are provided observations from previous rounds of comparing potential implementations against potential test suites, as well as a novel test suite.

**1. Analyze and Plan**
 - Carefully reflect on the provided observations.
 - Review the provided test suite.
 - Your new implementation must address the observations, pass all test cases, and adhere to the original task description.

**2. Write the implementation**
 - Include all necessary imports for external libraries.
 - Do not include a docstring.
 - Write the program in a Python code block marked with triple backticks (```).
\end{lstlisting}

\subsection{Feedback compression}\label{sec:compressor_description}
We consider two strategies for compressing the test reports. With \textit{insight reformulation}, we not only compress each report into actionable insights, but also reformulate old insights to avoid insight duplication and to ensure the context size only increases if new information is acquired. In contrast, with \textit{summary concatenation}, we simply ask the LLM to compress the test report of each round into a summary, which are then concatenated in the prompt. This ensures no information is accidentally lost, but results in a context that grows linearly with the number of rounds of in-context Thompson sampling. \looseness -1

\textbf{Insight reformulation}
\begin{lstlisting}
You are an expert Python programmer.
You are provided a task description, a candidate program implementation, a candidate test suite with feedback on the implementation, and additional past observations.
Your task is to first analyze the test feedback and then state an exhaustive list of insights that acts as a compressed history of all observations.

**1. Analyze Feedback**
 - Carefully review the program, tests, and feedback. This is your new data.
 - Identify errors in the program implementation, as well as weaknesses in the tests themselves.
 - Summarize all your findings. Be precise and do not speculate.

**2. State Insights**
 - Restate all insights from past observations as well as additional insights from the feedback analysis.
 - Mark each insight with `<observation> </observation>` tags.
 - Ensure each insight is self-contained, i.e., informative without considering the provided program and test suite.
\end{lstlisting}

\textbf{Summary concatenation}
\begin{lstlisting}
You are an expert Python programmer.
You are provided a task description, a candidate program implementation, and a candidate test suite with feedback on the implementation. Your task is to first analyze the test feedback and then summarize the implementation, tests, and feedback.

**1. Analyze Feedback**
 - Carefully review the program, tests, and feedback. This is your data.
 - Identify errors in the program implementation, as well as weaknesses in the tests themselves.

**2. State Summary**
 - Summarize the provided implementation, test suite, and all your findings.
 - Mark the summary with `<observation> </observation>` tags.
\end{lstlisting}

\newpage
\section{Proofs}\label{sec:proofs}

\subsection{Discussion of assumptions}\label{sec:discussion_assumptions}
Across this work, several assumptions are adopted. Here, we reassure the reader of their validity.

\paragraph{Non-empty test suites} In Definition~\ref{def:functional_code_similarity}, we tacitly assume all test suites from the test manifold to contain at least one test. This is without loss of generality, since if no test can be extracted from the language model output, we just consider the test suite to consist of a single test which always fails.

\paragraph{Almost surely unique correct behavior} In Proposition~\ref{prop:additivity_probability_of_correctness}, we assume that the true test suite is such that only one equivalence class of codes passes it. This is a desired property and necessary for program well-specification, and can thus be assumed true for the underlying generative process of Figure~\ref{fig:probabilistic_model}. However, it may not hold for LLM generated tests. As such, Proposition~\ref{prop:additivity_probability_of_correctness} should be understood as holding for the true joint code-test manifold, not for the LLM approximation. Note that Proposition~\ref{prop:additivity_probability_of_correctness} stands by itself and does not affect any of our other results.

\paragraph{LLM training approximately learns in-context Thompson sampling} Since pre-training teaches the LLM to approximate the Thompson sampling policy, the regret bound of Theorem~\ref{thm:reward_bound} directly maps to in-context backprompting of execution feedback. We remark, however, that post-training reshapes the learned code-test manifold and may result in behavior that deviates from the theory. Consequently, we suggest future work to investigate and contrast the in-context learning capabilities and black-box optimization capabilities of pre-trained models with those of post-trained models.

\paragraph{Subgaussian rewards and finite LLM response length} Theorem~\ref{thm:reward_bound} assumes subgaussian $r_{obs}$ and finite $|\mathrm{dom}(r)|$. The former is trivially satisfied since test suite pass-rates are within $[0,1]$ and since all bounded random variables are subgaussian. The latter holds as, due to memory limits in the Transformer architecture~\cite{vaswani2017}, the length of LLM output responses is always limited.

\paragraph{Strong assumptions for predicting the reward curves} At the end of Section~\ref{sec:regret_bounds}, we estimate the shape of empirical reward curves under backprompting. We stress that this is not a formal statement and the predicted curves require strong assumptions, in particular, tight regret bounds, a saturated cumulative uncertainty $\Gamma_T \approx \Gamma$, and an accurate truncated Taylor expansion $\sqrt{T} \text{-} \sqrt{T\text{-}1} \approx T^{\, \text{-}1/2}/2$.

\subsection{Proof sketches}
\label{sec:proof_sketches}

To improve readability, we provide brief sketches of the intuition behind our main theoretical results before presenting all formal proofs.

\textbf{Proof sketch of Proposition~\ref{prop:properties_s_similarity} (Properties of $s$-similarity).}
To prove that $s$-similarity is a valid kernel, we first show it holds for $s=1$ by representing the test outputs as one-hot vectors; the proportion of passed tests is then just the expectation of their inner products. Because the class of positive semi-definite kernels is closed under point-wise multiplication, multiplying the $s=1$ kernel by itself $s$ times ensures it remains a valid kernel for any $s \in \mathbb{N}$. Finally, the limit as $s \to \infty$ drives any fractional similarity strictly to $0$, leaving only perfect matches (value $1$), which aligns exactly with the definition of functional equivalence.

\textbf{Proof sketch of Proposition~\ref{prop:equivalence_classes} (Functional equivalence classes).}
Strict functional equivalence ($\mathrm{sim}_{p,e}^\infty(c_1, c_2) = 1$) requires two implementations $c_1, c_2$ to yield identical outputs on all test cases of non-zero measure. Because equality is an equivalence relation (it is reflexive, symmetric, and transitive), the $\mathrm{sim}_{p,e}^\infty$ kernel inherits these properties. Consequently, it partitions the space of all possible programs $\mathcal V^{\le L}$ into disjoint equivalence classes.

\textbf{Proof sketch of Theorem~\ref{thm:measure_smoothing} (Measure smoothing inductive bias).}
Because $s$-similarity is a positive definite kernel, we can embed the implementations into a Hilbert space using a feature map. The probability of a fuzzy neighborhood can then be expressed as the inner product between the code's feature representation and the mean embedding of the distribution. By applying the Cauchy-Schwarz inequality, the difference in probability mass between two highly similar codes is bounded by their geometric distance in the feature space, which directly links to the $s$-similarity.

\textbf{Proof sketch of Theorem~\ref{thm:consistency_and_bias} (Consistency and bias).}
For strong consistency, we invoke Kolmogorov's strong law of large numbers on the bounded, independent test suite evaluations, followed by the continuous mapping theorem to handle the exponent $s$. For bias, we rely on the linearity of expectation to show the estimator is perfectly unbiased for $s=1$. However, because $g(x) = x^s$ is strictly convex for $s > 1$, Jensen's inequality guarantees the estimator is biased upwards for all $s > 1$.

\textbf{Proof sketch of Theorem~\ref{thm:snr_dominance} (SNR dominance).}
We compare the variances of the two estimators. The soft estimator ($s=1$) is an average of independent bounded variables, so its variance shrinks inversely with the number of test suites $m$. In contrast, the hard estimator ($s \to \infty$) acts as a strict boolean AND across all sampled test suites, meaning its probability of success decays exponentially with $m$. Taking the ratio of their signal-to-noise ratios and applying a geometric series bound yields the $m^2$ dominance factor.

\textbf{Proof sketch of Proposition~\ref{prop:additivity_probability_of_correctness} (Additivity of probability of correctness).}
Under the assumption that there is almost surely only one functionally correct equivalence class, the success events for distinct classes are mutually exclusive. Therefore, the probability of the union of their success events (pass@$k$) is simply the sum of their individual success probabilities (pass@$1$). Maximizing this sum over $k$ selections trivially reduces to greedily picking the $k$ classes with the highest pass@1.

\textbf{Proof sketch of Theorem~\ref{thm:reward_bound} (Bayesian regret bound).}
We first decompose the reward into an observable component (execution results) and an unobservable component (task relevancy). By taking conditional expectations, we map the regret on the total reward function to a regret purely on the observable component. We then apply standard subgaussian tail bounds and union bounds (via Lemma~\ref{lem:uniform_bound}) to limit the expected maximum deviation. Finally, applying Cauchy-Schwarz over the $T$ steps bounds the cumulative regret by the sum of the squared subgaussian parameters.

\textbf{Proof sketch of Lemma~\ref{lem:uniform_bound}} relies on subgaussian tail integration and union bounds to limit the expected maximum of a set of random variables, producing the stated $\beta$ constant for the regret bound.\looseness -1

\textbf{Proof sketch of Lemma~\ref{lem:strict_expected_subgaussian_variance}} utilizes the law of total variance. Strict subgaussianity combined with the assumption of non-zero correlation forces the variance of the conditional mean to be strictly positive, proving the marginal subgaussian parameter strictly bounds the expected conditional one.

\textbf{Proof sketch of Lemma~\ref{lem:information_gain_for_gaussians_noiseless_gamma}} bounds the exact noiseless predictive variance by comparing it to a virtual process with noisy observations. This stochastic relaxation allows us to directly apply standard maximal information gain bounds for noisy Gaussian processes to upper-bound the cumulative variance.\looseness-1

\subsection{Main results}

\begin{proposition*}[\ref{prop:properties_s_similarity}]
    The similarity $\mathrm{sim}_{p,e}^s(c_1, c_2)$ is a positive semi-definite kernel with $\mathrm{sim}_{p,e}^s \geq 0$ and $\mathrm{sim}_{p,e}^s(c,c) = 1$. Moreover, $\mathrm{sim}^\infty_{p,e} (c_1,c_2) = \lim_{s \to \infty} \mathrm{sim}^s_{p,e} (c_1,c_2)$.
\end{proposition*}

\begin{proof}
    We first show that $\mathrm{sim}_{p,e}^1$ is a positive semi-definite kernel. Recall Definition~\ref{def:functional_code_similarity}:
    \[
    \textstyle \mathrm{sim}_{p,e}^1(c_1, c_2) := \mathbb E_{t\sim p}[ \frac{1}{|t|}\sum_{k=1}^{|t|} \mathds{1}_{O_k(c_1, t \vert e)=O_k(c_2, t \vert e)}]
    \]
    Let $\mathbb{O}$ be the finite set of possible test outputs, i.e., $O_k(c,t|e) \in \mathbb O\ \forall c,t,k$ under the specified environment $e \in \mathcal E$. For a fixed test $t$ and index $k$, define the feature map $\phi_{t,k}: \mathcal V^{\le L} \to \{0,1\}^{|\mathbb{O}|}$ such that $\phi_{t,k}(c)$ is a one-hot vector corresponding to the output $O_k(c, t | e)$. The indicator function can then be written as an inner product:
    \[
    \mathds{1}_{O_k(c_1, t | e) = O_k(c_2, t | e)} = \langle \phi_{t,k}(c_1), \phi_{t,k}(c_2) \rangle.
    \]
    Since the inner product is a valid kernel, and the class of kernels is closed under addition and scaling by non-negative constants (linearity of expectation), $\mathrm{sim}_{p,e}^1(c_1, c_2)$ is a valid kernel. Concerning non-negativity: The indicator function is non-negative, and the probability measure $p(t)$ is non-negative, thus $\mathrm{sim}_{p,e}^1(c_1, c_2) \ge 0$. Finally, the kernel is normalized since
    \begin{equation*}
    \textstyle \mathrm{sim}_{p,e}^1(c, c) = \mathbb{E}_{t \sim p} [ \frac{1}{|t|} \sum_{k=1}^{|t|} \mathds{1}_{O_k(c, t | e) = O_k(c, t | e)} ] = \mathds{E}_{t \sim p} [1] = 1.
    \end{equation*}
    To extend these results to $\mathrm{sim}_{p,e}^s$ for $s \not = 1$, recall that the class of positive semi-definite kernels is closed under point-wise multiplication (as a consequence of the Schur Product Theorem). Since $\mathrm{sim}_{p,e}^s(c_1, c_2) = (\mathrm{sim}_{p,e}^1(c_1, c_2))^s$ is the point-wise product of the kernel $\mathrm{sim}_{p,e}^1$ with itself $s$ times, $\mathrm{sim}_{p,e}^s$ is a kernel for any $s \in \mathbb{N}$. 
    
    Finally, to prove convergence to functional equivalence, 
    consider two cases: If $\mathrm{sim}_{p,e}^\infty(c_1, c_2) = 1$ then $O_k(c_1, t \vert e) = O_k(c_2, t \vert e)\ \forall k \in [|t|], t \in \mathcal V^{\le L} : p(t) > 0$. Therefore, $\mathrm{sim}_{p,e}^1(c,c) = 1$ and thus $\mathrm{sim}_{p,e}^s(c,c) = 1\ \forall s \in \mathbb N$, including in the limit. In contrast, if $\mathrm{sim}_{p,e}^\infty(c_1, c_2) \not= 1$, then  $\mathrm{sim}_{p,e}^\infty(c_1, c_2) = 0$ and $\exists t \in \mathcal V^{\le L}: p(t) > 0$ and $O_k(c_1, t \vert e) \not = O_k(c_2, t \vert e)$ for some $k \in [|t|]$. But then $\mathrm{sim}_{p,e}^1(c_1,c_2) \in [0,1)$ and thus $\lim_{s\to\infty}\mathrm{sim}_{p,e}^s(c_1,c_2) = \lim_{s\to\infty} (\mathrm{sim}_{p,e}^1(c_1,c_2))^s  = 0$. In either case, functional $s$-similarity converges to functional equivalence as $s \to \infty$.
\end{proof}

\begin{proposition*}[\ref{prop:equivalence_classes}]
  Define $\mathcal N_c^\infty = \{c^\prime \in \mathcal V^{\le L} : \mathrm{sim}_{p,e}^\infty(c, c^\prime) = 1 \}$. Then $\{\mathcal N_c^\infty : c \in \mathcal V^{\le L}\}$ is a set of equivalence classes partitioning $\mathcal V^{\le L}$.  
\end{proposition*}

\begin{proof}
    By Definition~\ref{def:functional_equivalence}, $\mathrm{sim}_{p,e}^\infty(c, c^\prime) = 1$ if and only if $c, c^\prime$ produce identical outputs on all test cases with non-zero support. Since the equality of outputs is reflexive, symmetric, and transitive, $\mathrm{sim}_{p,e}^\infty$ inherits these properties, partitioning the code space into disjoint equivalence classes $\mathcal N_c^\infty$.
\end{proof}

\begin{theorem*}[\ref{thm:measure_smoothing}]
    Let $\mathrm{sim}_{p,e}^s(c,c^\prime) \geq 1-\varepsilon$ for $\varepsilon \in (0,1)$. Then 
    \begin{equation*}
    |p(\mathcal N_c^s) - p(\mathcal N_{c^\prime}^s)| \leq \sqrt{2\varepsilon}.
    \end{equation*}
\end{theorem*}
\begin{proof}
    Since $\mathrm{sim}_{p,e}^s$ is a positive definite kernel with $\mathrm{sim}_{p,e}^s(c,c)=1$ (see Proposition~\ref{prop:properties_s_similarity}), there is a feature map $\phi: \mathcal V^{\le L} \to \mathcal H$ into a Hilbert space s.t. $\mathrm{sim}_{p,e}^s(x,y) = \langle \phi(x), \phi(y) \rangle$ and $\|\phi(x)\|=1$.\looseness-1
    
    As a result, we can express the probability of the fuzzy neighborhood $p(\mathcal N_c^s)$ as an inner product with the \textit{mean embedding} of the distribution $p$, which is defined as $\mu_p := \mathbb{E}_{x \sim p}[\phi(x)]$. It holds that
    \begin{align*}
    p(\mathcal N_c^s) & = \mathbb{E}_{x \sim p}[\mathrm{sim}_{p,e}^s(c, x)]\\ 
    & = \mathbb{E}_{x \sim p}[\langle \phi(c), \phi(x) \rangle] = \langle \phi(c), \mu_p \rangle.
    \end{align*}
    
    The difference in probability mass is then governed by the distance in feature space, i.e., we have that
    \begin{align*}
        |p(\mathcal N_c^s) - p(\mathcal N_{c^\prime}^s)| 
        &= |\langle \phi(c), \mu_p \rangle - \langle \phi(c^\prime), \mu_p \rangle| \\
        &= |\langle \phi(c) - \phi(c^\prime), \mu_p \rangle|.
    \end{align*}
    Applying the Cauchy-Schwarz inequality then yields the following inequality:
    $$|\langle \phi(c) - \phi(c^\prime), \mu_p \rangle| \leq \|\phi(c) - \phi(c^\prime)\| \cdot \|\mu_p\|.$$
    
    We bound the first term by quadratic expansion of the squared norm, using that $\|\phi(c)\| = 1$ and that by assumption $\mathrm{sim}_{p,e}^s(c,c^\prime) \geq 1-\varepsilon$. We get
    \begin{align*}
    \|\phi(c) - \phi(c^\prime)\|^2 &= \langle \phi(c)-\phi(c^\prime), \phi(c)-\phi(c^\prime) \rangle \\
    &= \|\phi(c)\|^2 + \|\phi(c^\prime)\|^2 - 2\langle \phi(c), \phi(c^\prime) \rangle \\
    &= 2 - 2 \mathrm{sim}_{p,e}^s(c, c^\prime) \leq 2 - 2(1-\varepsilon) = 2\varepsilon.
    \end{align*}
    Thus, $\|\phi(c) - \phi(c^\prime)\| \leq \sqrt{2\varepsilon}$. For the second term, by Jensen's inequality, the norm of the mean embedding is bounded by the expected norm:
    $$\|\mu_p\| = \|\mathbb{E}_{x}[\phi(x)]\| \leq \mathbb{E}_{x}[\|\phi(x)\|] = 1.$$
    
    Multiplying these bounds gives $|p(\mathcal N_c^s) - p(\mathcal N_{c^\prime}^s)| \leq \sqrt{2\varepsilon} \cdot 1$.
\end{proof}

\begin{theorem*}[\ref{thm:consistency_and_bias}] 
    $\widehat{\mathrm{sim}}_{p,e;m}^s$ is a strongly consistent estimator of $\mathrm{sim}_{p,e}^s$, i.e., 
    \begin{equation*}
        \mathbb P[\lim\nolimits_{m \to \infty} \widehat{\mathrm{sim}}_{p,e;m}^s(c_1, c_2) = \mathrm{sim}_{p,e}^s (c_1, c_2)] = 1\ \forall c_1, c_2 \in \mathcal V^{\le L}.
    \end{equation*}
    Yet, it is biased unless $s=1$.
\end{theorem*}

\begin{proof}
    Let us first define the proportion of passed test cases on a single sampled test suite $t_j$ as a random variable $Z_j$: 
    \begin{equation*}     
        \textstyle Z_j := \frac{1}{|t_j|} \sum_k \mathds{1}_{O_k(c_1, t_j \vert e) = O_k(c_2, t_j \vert e)} 
    \end{equation*}
    Because the test suites $t_j \sim p(t)$ are sampled i.i.d., the terms $Z_1, \dots, Z_m$ are i.i.d. bounded random variables on $[0, 1]$. By definition, their expected value is the true $1$-similarity:
    \begin{equation*}     
        \mu := \mathbb{E}[Z_j] = \mathrm{sim}_{p,e}^1(c_1, c_2) \end{equation*}  We can rewrite the true $s$-similarity and its Monte Carlo estimator in terms of $Z$: 
    \begin{itemize}
        \item True value: $\mathrm{sim}_{p,e}^s(c_1, c_2) = \mu^s$
        \item Estimator: $\widehat{\mathrm{sim}}_{p,e;m}^s(c_1, c_2) = (\bar{Z}_m)^s$, where $\bar{Z}_m = \frac{1}{m} \sum_{j=1}^m Z_j$ is the sample mean.
    \end{itemize}
    
    \underline{Part 1: Strong Consistency}
    
    To prove strong consistency, we must show that the estimator converges almost surely to the true value as the sample size $m \to \infty$.  \begin{enumerate}
        \item Since $Z_j$ are i.i.d. with finite expectation $\mu$, we can apply Kolmogorov's strong law of large numbers, which dictates that the sample mean converges almost surely to the true mean:     \begin{equation*}
            \mathbb{P} \left[ \lim_{m \to \infty} \bar{Z}_m = \mu \right] = 1
        \end{equation*}
        \item The function $g(x) = x^s$ is continuous for all $x \in [0, 1]$ and $s \in \mathbb{N}$.
        \item By the continuous mapping theorem, if a sequence of random variables $\bar{Z}_m \xrightarrow{a.s.} \mu$, then $g(\bar{Z}_m) \xrightarrow{a.s.} g(\mu)$. Therefore:
    \begin{equation*}         \mathbb{P} \left[ \lim_{m \to \infty} (\bar{Z}_m)^s = \mu^s \right] = 1
    \end{equation*}
    \end{enumerate}
    Substituting back our definitions yields: \begin{equation*}
        \mathbb{P} \left[ \lim_{m \to \infty} \widehat{\mathrm{sim}}_{p,e;m}^s(c_1, c_2) = \mathrm{sim}_{p,e}^s (c_1, c_2) \right] = 1
    \end{equation*} 
    This concludes the proof of strong consistency. 
    
    \underline{Part 2: Bias}
    
    To evaluate bias, we must check if the expected value of the estimator equals the true value, i.e., for which $s \in \mathbb N$ it holds that $\mathbb{E}[(\bar{Z}_m)^s] = \mu^s$.
    
    \textit{Case 1: $s = 1$}
    
    By the linearity of expectation, the expected value of the sample mean is exactly the population mean: \begin{equation*}
        \mathbb{E}[(\bar{Z}_m)^1] = \mathbb{E}[\bar{Z}_m] = \mu = \mu^1
    \end{equation*}
    Thus, for $s = 1$, the estimator is perfectly unbiased.
    
    \textit{Case 2: $s > 1$} 
    
    For any $s > 1$, the function $g(x) = x^s$ is strictly convex on $[0, 1]$. Now, for $Z_j$ not almost surely constant, we can apply Jensen's inequality for strictly convex functions, which states that the expectation of a convex function is strictly greater than the convex function of the expectation:
    \begin{equation*}
        \mathbb{E}[(\bar{Z}_m)^s] > (\mathbb{E}[\bar{Z}_m])^s = \mu^s
    \end{equation*}
    Because the expected value of the estimator is strictly greater than the true parameter for all non-degenerate test distributions, the estimator is biased upwards for all $s > 1$. 
\end{proof}

\begin{theorem*}[~\ref{thm:snr_dominance}]
    Let $\mathrm{sim}_{p,e}^1(c_1, c_2) = \mu \in (0,1)$ be the true functional similarity.
    We define the Signal-to-Noise Ratio (SNR) of an estimator $\widehat{X}$ as $\mathrm{SNR}(\widehat{X}) := (\mathbb{E}[\widehat{X}])^2 / \mathrm{Var}(\widehat{X})$.
    For $m \ge 1$, the SNR of the soft estimator dominates the hard estimator by a factor of at least $m^2$:
    \begin{equation*}
    \mathrm{SNR}(\widehat{\mathrm{sim}}_{p,e;m}^1(c_1, c_2))\,/\,\mathrm{SNR}(\widehat{\mathrm{sim}}_{p,e;m}^\infty(c_1, c_2)) \geq m (1/{\mu})^{m-1} \cdot \tfrac{1-\mu^m}{1-\mu} \ge m^2
    \end{equation*}
\end{theorem*}
\begin{proof}
    Let $Z := \widehat{\mathrm{sim}}_{p,e;1}^1(c_1, c_2 \, \vert \, e)$ be the random variable representing the similarity score on a single test suite, i.e., $Z \in [0,1]$ and $\mathbb{E}[Z] = \mu \in (0,1)$.
    The soft estimator is the sample mean of $m$ i.i.d. draws of $Z$. Its SNR depends on the variance of $Z$:
    $$\mathrm{SNR}(\widehat{\mathrm{sim}}_{p,e;m}^1) = \frac{\mathbb{E}[\widehat{\mathrm{sim}}_{p,e;m}^1]^2}{\mathrm{Var}(\widehat{\mathrm{sim}}_{p,e;m}^1)} = \frac{\mu^2}{\mathrm{Var}(Z)/m} = \frac{m\mu^2}{\mathrm{Var}(Z)}.$$
    Since $Z$ is bounded in $[0,1]$, $\mathrm{Var}(Z) = \mathbb E[Z^2] - \mu^2 \leq \mathbb E[Z] - \mu^2 = \mu(1-\mu)$. Thus:
    $$\mathrm{SNR}(\widehat{\mathrm{sim}}_{p,e;m}^1) \ge \frac{m\mu^2}{\mu(1-\mu)} = \frac{m\mu}{1-\mu}.$$
    
    The hard estimator is an indicator that is $1$ if and only if all sampled suites are perfect matches, i.e., $\widehat{\mathrm{sim}}_{p,e;m}^\infty = \prod_{j=1}^m \mathds{1}_{Z_j=1}$.
    Consequently, it is a Bernoulli variable with parameter $p_1^m$, where $p_1 = \mathbb P[Z=1]$. As a result,
    $$\mathrm{SNR}(\widehat{\mathrm{sim}}_{p,e;m}^\infty) = \frac{(p_1^m)^2}{p_1^m(1-p_1^m)} = \frac{p_1^m}{1-p_1^m}.$$
    Now, since $Z \geq 0$, $p_1 \le \mathbb{E}[Z] = \mu$.
    Moreover, for $m \geq 1$, $f(x) = \frac{x}{1-x}$ and $g(x) = x^m$ are both monotonously increasing between $0$ and $1$. As a result, 
    $$\mathrm{SNR}(\widehat{\mathrm{sim}}_{p,e;m}^\infty) \le \frac{\mu^m}{1-\mu^m}.$$
    
    Taking the ratio of the lower bound of the smooth SNR to the upper bound of the sharp SNR and applying the formula for a finite geometric series in reverse yields
    $$ \frac{\mathrm{SNR}(\widehat{\mathrm{sim}}_{p,e;m}^1)}{\mathrm{SNR}(\widehat{\mathrm{sim}}_{p,e;m}^\infty)} \geq \frac{\frac{m\mu}{1-\mu}}{\frac{\mu^m}{1-\mu^m}} = \frac{m}{\mu^{m-1}} \frac{1-\mu^m}{1-\mu} = \frac{m}{\mu^{m-1}} \sum_{k=0}^{m-1} \mu^{k} = m \sum_{k=0}^{m-1} \mu^{k-(m-1)}.$$
    Finally, we observe that since $\mu \in (0,1)$, every term $\mu^{k-(m-1)} \ge 1$. As there are $m$ such terms:
    $$ \frac{\mathrm{SNR}(\widehat{\mathrm{sim}}_{p,e;m}^1)}{\mathrm{SNR}(\widehat{\mathrm{sim}}_{p,e;m}^\infty)} \ge m \cdot m = m^2. $$
\end{proof}

\begin{proposition*}[\ref{prop:additivity_probability_of_correctness}]
    Suppose that the true test suite $t \in \mathcal V^{\le L}$ verifying an implementation $c \in \mathcal V^{\le L}$ in environment $e\in \mathcal E$ is almost surely such that $\exists !$ correct $\mathcal N_c^\infty$. Define $\mathrm{pass@k}(\{\mathcal N_{c_i}^\infty\}_{i=1}^k) := \mathbb P_{t\sim p} [\exists i \in \{1,\ldots,k\} : R(c_i,t\, \vert\, e)=1]$. Then, $\mathrm{pass@k}$ is maximized by greedily selecting the $k$ code classes $\{\mathcal N_{c_i}^\infty\}_{i=1}^k$ with largest $\mathrm{pass@1}$.
\end{proposition*}

\begin{proof}
    Let $S = \{\mathcal N_{c_1}^\infty, \ldots, \mathcal N_{c_k}^\infty\}$ be the set of $k$ selected equivalence classes. We aim to maximize pass@k:
    \begin{equation*}
        \mathrm{pass@k}(S) := \mathbb P_{t \sim p} \big[ \exists\, \mathcal N_c^\infty \in S : R(c, t \, \vert \, e) = 1 \big].
    \end{equation*}
    Define the success event for a class as $E_c := \{t \in \mathcal V^{\le L} : R(c, t \, \vert \, e) = 1\}$. 
    Since by assumption there exists almost surely one correct equivalence class, for any two distinct classes $\mathcal N_{c_i}^\infty, \mathcal N_{c_j}^\infty \in S$, the intersection of their success events has zero measure ($\mathbb P[E_{c_i} \cap E_{c_j}] = 0$).
    
    This mutual exclusivity allows us to decompose the probability of the union into a sum:
    \begin{equation*}
        \mathrm{pass@k}(S) = \mathbb P_{t\sim p} \big[ \bigcup_{\mathcal N_c^\infty \in S} E_c \big] = \sum_{\mathcal N_c^\infty \in S} \mathbb P_{t\sim p} [E_c] = \sum_{\mathcal N_c^\infty \in S} \mathrm{pass@1}(\{\mathcal N_c^\infty\}).
    \end{equation*}
    Maximizing this sum subject to the constraint $|S|=k$ is a standard selection problem, for which the optimal solution is the greedy strategy: selecting the $k$ classes with the largest $\mathrm{pass@1}$.
\end{proof}

\begin{theorem*}[\ref{thm:reward_bound}]
Decompose the reward into independent components $r := r_{obs} + r_{hid}$, where $\mathbb E[r_{hid}]$ is finite, and $r_{obs}(x)$ is $\sigma_n(x)$-subgaussian given observation history $\mathcal H_n := (x_i, r_{obs}(x_i))_{i=1}^{n-1}$. Define the Thompson sampling policies $x_{obs}^* \sim \mathbb P[x^* \,\vert\, r_{obs}]$ and $x_n \sim \mathbb P[x^* \,\vert\, \mathcal H_n]$ for $x^* := \arg\max_x r(x)$. Then
\begin{equation*}
    \textstyle \mathbb E[\sum_{n=1}^T r(x^*_{obs}) - r(x_n)] \leq \beta\sqrt{T \cdot \Gamma_T}
\end{equation*}
for $\beta := \sqrt{2 \log(|\mathrm{dom}(r)|)} + 3$ and cumulative uncertainty $\Gamma_T := \mathbb E[\sum_{n=1}^T\!\sigma_n^2(x_n)]$.
\end{theorem*}

\begin{proof}
Without loss of generality $\mathbb E[r_{hid}] = 0$, since nonzero expectations can be absorbed into a shifted $r_{obs}$. Since $r_{hid} \perp r_{obs}$, it follows that $r_{hid} \perp x_{obs}^*, x_n$ and therefore
\begin{equation*}
    \mathbb E[r(x^*_{obs}) - r(x_n) \mid \mathcal H_n] = \mathbb E[r_{obs}(x^*_{obs}) - r_{obs}(x_n) \mid \mathcal H_n].
\end{equation*}
Conditioned on $\mathcal H_n$, define $(r_{obs}^\prime, r_{hid}^\prime) \overset{d}{:=} (r_{obs}, r_{hid})$ as an independent copy of the true reward tuple and $r_{hid}^{\prime\prime}\mid r_{obs} \overset{d}{:=}r_{hid} \mid r_{obs}$ as an independent copy of $r_{hid}$ given full knowledge of $r_{obs}$. Then it holds that $x^*_{obs} := \arg\max_x r_{obs}(x) + r_{hid}^{\prime\prime}(x)$ and $x_n := \arg\max_x r_{obs}^\prime(x) + r_{hid}^\prime(x)$, and therefore $\mathbb E[r_{obs}(x_{obs}^*) \mid \mathcal H_n] = \mathbb E[r_{obs}^\prime(x_n) \mid \mathcal H_n]$. As a result,
\begin{equation*}
    \mathbb E[r_{obs}(x^*_{obs}) - r_{obs}(x_n) \mid \mathcal H_n] = \mathbb E[r_{obs}^\prime(x_n)-r_{obs}(x_n) \mid \mathcal H_n].
\end{equation*}
Next, we take expectations w.r.t. $\mathcal H_n$ and apply Lemma~\ref{lem:uniform_bound} and Jensen's inequality to obtain
\begin{align*}
    \mathbb E[r(x^*_{obs}) - r(x_n)] = \mathbb E[r_{obs}^\prime(x_n)-r_{obs}(x_n)] & = \mathbb E[\mathbb E[r_{obs}^\prime(x_n)-r_{obs}(x_n) | \mathcal H_n]]\\ & \leq \mathbb E[\beta\sqrt{\mathbb E[\sigma_{n}^2(x_n) | \mathcal H_n]}] \leq \beta \sqrt{\mathbb E[\sigma_{n}^2(x_n)]}.
\end{align*}
Finally, we apply Cauchy-Schwarz on the full regret sequence
\begin{equation*}
    \textstyle \mathbb E[\sum_{n=1}^T r(x^*_{obs}) - r(x_n)]
    \leq \beta \sum_{n=1}^T \sqrt{\mathbb E[\sigma_{n}^2(x_n)]}
     \leq \textstyle \beta \sqrt{T \sum_{n=1}^T \mathbb E[\sigma_{n}^2(x_n)]}.
\end{equation*}  
\end{proof}

\subsection{Lemmas}

\begin{lemma}[]\label{lem:uniform_bound}
    Let $r, r^\prime$ be i.i.d., element-wise $\sigma_x$-subgaussian, and let $x = f(r) \in \{1, \ldots, |\mathrm{dom}(r)|\}$ for any function $f$. Then 
    \begin{align*}
        \mathbb E[|r_x-\mu_x|] & \leq \sqrt{\mathbb E[\sigma_{x}^2] \cdot (2 \log (2|\mathrm{dom}(r)|)+2)} \text{ and }\\
        \mathbb E[|r_x-r^\prime_x|] & \leq \beta \cdot \sqrt{\mathbb E[\sigma_{x}^2]}.
    \end{align*}
    for $\beta := \sqrt{2 \log(|\mathrm{dom}(r)|)} + 3$.
\end{lemma}
\begin{proof}
    Define $Z_z = (r_z - \mu_z) / \sigma_z$. Since $r_z - \mu_z$ is $\sigma_z$-subgaussian, $Z_z$ is $1$-subgaussian. Cauchy-Schwarz then gives, $\mathbb E[|r_x - \mu_x|] = \mathbb E [\sigma_{x} | Z_x|] \leq \sqrt{\mathbb E[\sigma_{x}^2] \mathbb E[Z_x^2]}$. Now, $Z_x^2 \leq 
    \max_{1 \leq j \leq |\mathrm{dom}(r)|} Z_j^2$, so the tail integral formula for expectations with union bound gives $\textstyle \mathbb E[Z_x^2] \leq \mathbb E[\max_{1 \leq j \leq |\mathrm{dom}(r)|} Z_j^2] = \int_0^\infty\! \mathbb P[\max_{1\leq j \leq |\mathrm{dom}(r)|} Z_j^2 \geq t] dt \leq \int_0^\infty \!\min(1, 2|\mathrm{dom}(r)| e^{-t/2}) dt =  2 \log (2|\mathrm{dom}(r)|)+2$. The first result then follows immediately. For the second result, apply the triangle inequality $\mathbb E[|r_x-r^\prime_x|] \leq \mathbb E[|r_x-\mu_x|] + \mathbb E[|r_x^\prime-\mu_x|]$, bound the first term as above, and the second by $\sqrt{\mathbb E[\sigma_{x}^2]}$, using that $r^\prime \perp x = f(r)$ and thus $\mathbb E[(Z_x^\prime)^2] \leq 1$. This yields an overall factor of $1 + \sqrt{2 \log(2|\mathrm{dom}(r)|) + 2}$. Applying the subadditivity of the square root ($\sqrt{A+B} \leq \sqrt{A} + \sqrt{B}$), we can bound this factor by $1 + \sqrt{2 \log 2 + 2} + \sqrt{2 \log(|\mathrm{dom}(r)|)}$. Since $1 + \sqrt{2 \log 2 + 2} \approx 2.84 \leq 3$, we safely obtain the upper bound using $\beta$.
\end{proof}

\begin{lemma}\label{lem:strict_expected_subgaussian_variance}
    Let $X$ and $Y$ be correlated random variables where $X$ and $X\mid Y$ are strictly subgaussian. Let $\sigma^2_X$ denote the optimal subgaussian parameter of $X$, and $\sigma^2_{X|Y}$ denote the optimal conditional subgaussian parameter of $X$ given $Y$. Then,
    \begin{equation*}
        \sigma^2_X > \mathbb{E}_Y[\sigma^2_{X|Y}].
    \end{equation*}
\end{lemma}

\begin{proof}
    By strict subgaussianity, the optimal subgaussian parameters are equal to the variance of the distributions. Therefore, the marginal parameter is $\sigma^2_X = \mathrm{Var}(X)$, and the conditional parameter satisfies $\sigma^2_{X|Y} = \mathrm{Var}(X \mid Y)$. By the Law of Total Variance, the marginal variance decomposes as:
    \begin{equation*}
        \mathrm{Var}(X) = \mathbb{E}_Y[\mathrm{Var}(X \mid Y)] + \mathrm{Var}(\mathbb{E}[X \mid Y]).
    \end{equation*}
    Next, note that the conditional mean $\mathbb{E}[X \vert Y]$ cannot be almost surely constant, otherwise $\mathbb{E}[X \vert Y] = \mathbb E[\mathbb{E}[X \vert Y]] = \mathbb E[X]$ a.s., and thus $\mathrm{Cov}(X,Y) := \mathbb E[XY] - \mathbb E[X] \mathbb E[Y] = \mathbb E[Y (\mathbb E[X \vert Y] - \mathbb E[X])] = 0$, contradicting the correlation assumption. As a result, we can apply the strict form of Jensen's inequality to infer that the variance of the conditional mean is strictly positive: $\mathrm{Var}(\mathbb E[X \vert Y]) = \mathbb E\left[(\mathbb E[X \vert Y])^2\right] - (\mathbb E[\mathbb E[X \vert Y]])^2 > 0$. Substituting the subgaussian parameters into the variance decomposition and applying this strict positivity yields the final bound:
    \begin{equation*}
        \sigma^2_X = \mathbb{E}_Y[\sigma^2_{X|Y}] + \mathrm{Var}(\mathbb{E}[X \mid Y]) > \mathbb{E}_Y[\sigma^2_{X|Y}].
    \end{equation*}
\end{proof}

\begin{lemma}\label{lem:information_gain_for_gaussians_noiseless_gamma}
    Let $r \sim \mathcal N(\mu, \Sigma)$ be a multivariate Gaussian with uniformly bounded standard deviations $\sigma_1(x) \leq 1\ \forall x$ that is consecutively evaluated at locations $(x_t)_{t=1}^T$ with \textit{noiseless} observations $y_t = r(x_t)$. For any artificial noise level $\eta^2 > 0$, define the noisy maximal information gain $\gamma_T(\eta) := \max_{x_{1:T}} I(\tilde{y}_{1:T}; r)$, where $\tilde{y}_t = r(x_t) + \varepsilon_t$ and $\varepsilon_t \sim \mathcal N(0, \eta^2)$. Then, the aggregated predictive variances $\sigma_t^2(x_t) := \mathrm{Var}[r(x_t) \mid y_1, \ldots, y_{t-1}]$ satisfy
    \begin{equation*}
        \sum_{t=1}^T \sigma_t^2(x_t) \leq \frac{2\gamma_T(\eta)}{\log(1+\eta^{-2})}.
    \end{equation*}
\end{lemma}

\begin{proof}
    We consider a virtual process that observes the same locations $(x_t)_{t=1}^T$ but receives noisy observations $\tilde{y}_t = r(x_t) + \varepsilon_t$, with i.i.d. noise $\varepsilon_t \sim \mathcal N(0, \eta^2)$. Let $\tilde{\sigma}_t^2(x_t) := \mathrm{Var}[r(x_t) \mid \tilde{y}_1, \ldots, \tilde{y}_{t-1}]$ denote the predictive variance under this noisy process.
    
    Let $k_{t-1}$ be the prior covariance vector between $y_t$ and the previously evaluated locations $y_{1:t-1}$ (and also $\tilde y_t$ and $\tilde y_{1:t-1}$), and let $K_{t-1}$ be the prior covariance matrix of $y_{1:t-1}$. According to the closed-form Gaussian update rules, the exact, noiseless predictive variance is given by
    \begin{equation*}
        \sigma_t^2(x_t) = \sigma_1^2(x_t) - k_{t-1}^\top K_{t-1}^{-1} k_{t-1}.
    \end{equation*}
    In contrast, for the noisy virtual process, the covariance of the observations incorporates the independent noise variance $\eta^2$, yielding
    \begin{equation*}
        \tilde{\sigma}_t^2(x_t) = \sigma_1^2(x_t) - k_{t-1}^\top (K_{t-1} + \eta^2 I)^{-1} k_{t-1}.
    \end{equation*}
    Because $\eta^2 I$ is strictly positive definite, it follows that $(K_{t-1} + \eta^2 I)^{-1} \preceq K_{t-1}^{-1}$ in the positive semi-definite order. This implies $k_{t-1}^\top (K_{t-1} + \eta^2 I)^{-1} k_{t-1} \leq k_{t-1}^\top K_{t-1}^{-1} k_{t-1}$, meaning the subtracted reduction in variance is strictly smaller in the noisy case. Consequently, the noiseless predictive variance is upper-bounded by the noisy predictive variance:
    \begin{equation*}
        \sigma_t^2(x_t) \leq \tilde{\sigma}_t^2(x_t) \quad \forall t.
    \end{equation*}
    Intuitively, this confirms that conditioning on the noiseless values provides at least as much information as conditioning on the noisy ones.
    
     Because the standard deviation remains bounded ($\tilde{\sigma}_t(x_t) \leq \sigma_1(x_t) \leq 1$), we can directly apply the standard information gain bound for noisy Gaussian processes to the virtual process (\citealp{srinivas2009gaussian}, with the relevant statements combined into Lemma 4 of \citealp{menet2026thompson}). We obtain:
    \begin{equation*}
        \sum_{t=1}^T \sigma_t^2(x_t) \leq \sum_{t=1}^T \tilde{\sigma}_t^2(x_t) 
        \leq \frac{2\cdot I(\tilde{y}_{1:T}; r)}{\log(1+\eta^{-2})}.
    \end{equation*}
\end{proof}

%% file: checklist.tex
\newpage
\section*{NeurIPS Paper Checklist}

\begin{enumerate}

\item {\bf Claims}
    \item[] Question: Do the main claims made in the abstract and introduction accurately reflect the paper's contributions and scope?
    \item[] Answer: \answerYes{} % Replace by \answerYes{}, \answerNo{}, or \answerNA{}.
    \item[] Justification: The abstract and introduction accurately reflect the paper's scope, emphasizing the theoretical analysis of post-generation selection and in-context Thompson sampling.

\item {\bf Limitations}
    \item[] Question: Does the paper discuss the limitations of the work performed by the authors?
    \item[] Answer: \answerYes{} % Replace by \answerYes{}, \answerNo{}, or \answerNA{}.
    \item[] Justification: The paper includes a dedicated and exhaustive "Limitations" section in the appendix.

\item {\bf Theory assumptions and proofs}
    \item[] Question: For each theoretical result, does the paper provide the full set of assumptions and a complete (and correct) proof?
    \item[] Answer: \answerYes{} % Replace by \answerYes{}, \answerNo{}, or \answerNA{}.
    \item[] Justification: Full sets of assumptions and complete proofs for all theoretical claims (e.g., Propositions~\ref{prop:properties_s_similarity},~\ref{prop:equivalence_classes},~\ref{prop:additivity_probability_of_correctness} and Theorems~\ref{thm:measure_smoothing},~\ref{thm:consistency_and_bias},~\ref{thm:snr_dominance},~\ref{thm:reward_bound}) are provided in Appendix~\ref{sec:proofs}

    \item {\bf Experimental result reproducibility}
    \item[] Question: Does the paper fully disclose all the information needed to reproduce the main experimental results of the paper to the extent that it affects the main claims and/or conclusions of the paper (regardless of whether the code and data are provided or not)?
    \item[] Answer: \answerYes{} % Replace by \answerYes{}, \answerNo{}, or \answerNA{}.
    \item[] Justification: The paper fully discloses models, datasets, system prompts, and feedback compression techniques in Appendix~\ref{sec:experimental_details}. Moreover, the entire code base is made available.

\item {\bf Open access to data and code}
    \item[] Question: Does the paper provide open access to the data and code, with sufficient instructions to faithfully reproduce the main experimental results, as described in supplemental material?
    \item[] Answer: \answerYes{} % Replace by \answerYes{}, \answerNo{}, or \answerNA{}.
    \item[] Justification: The code was provided in the supplementary materials and is released publicly at \url{https://github.com/IBM/theory-of-coding-agents} under a permissive license.

\item {\bf Experimental setting/details}
    \item[] Question: Does the paper specify all the training and test details (e.g., data splits, hyperparameters, how they were chosen, type of optimizer) necessary to understand the results?
    \item[] Answer: \answerYes{} % Replace by \answerYes{}, \answerNo{}, or \answerNA{}.
    \item[] Justification: Hyperparameters (e.g., temperature 1.0, top-p 0.95) and other evaluation settings are clearly specified in Appendix~\ref{sec:extensive_model_description} and follow standard practices. Moreover, the code base has configuration files that do not leave any parameters unspecified.

\item {\bf Experiment statistical significance}
    \item[] Question: Does the paper report error bars suitably and correctly defined or other appropriate information about the statistical significance of the experiments?
    \item[] Answer: \answerYes{} % Replace by \answerYes{}, \answerNo{}, or \answerNA{}.
    \item[] Justification: The experimental tables and figures correctly report the mean and standard error over 5 random seeds. The latter is the correct statistical measure to determine statistical significance of our results.

\item {\bf Experiments compute resources}
    \item[] Question: For each experiment, does the paper provide sufficient information on the computer resources (type of compute workers, memory, time of execution) needed to reproduce the experiments?
    \item[] Answer: \answerYes{} % Replace by \answerYes{}, \answerNo{}, or \answerNA{}.
    \item[] Justification: The compute environment (AMD EPYC 7763 64-Core CPU, 8 NVIDIA A100 GPUs) and specific time requirements are detailed in Section~\ref{sec:runtimes} and Appendix~\ref{sec:experimental_details}.
    
\item {\bf Code of ethics}
    \item[] Question: Does the research conducted in the paper conform, in every respect, with the NeurIPS Code of Ethics \url{https://neurips.cc/public/EthicsGuidelines}?
    \item[] Answer: \answerYes{} % Replace by \answerYes{}, \answerNo{}, or \answerNA{}.
    \item[] Justification: The research utilizes standard benchmarking datasets and open-weight models, conforming to standard ethical guidelines for machine learning research.

\item {\bf Broader impacts}
    \item[] Question: Does the paper discuss both potential positive societal impacts and negative societal impacts of the work performed?
    \item[] Answer: \answerYes{} % Replace by \answerYes{}, \answerNo{}, or \answerNA{}.
    \item[] Justification: There is a dedicated impact statement at the beginning of the Appendix.
    
\item {\bf Safeguards}
    \item[] Question: Does the paper describe safeguards that have been put in place for responsible release of data or models that have a high risk for misuse (e.g., pre-trained language models, image generators, or scraped datasets)?
    \item[] Answer: \answerNA{} % Replace by \answerYes{}, \answerNo{}, or \answerNA{}.
    \item[] Justification: The paper does not release high-risk models or scraped datasets that require specific safeguards.

\item {\bf Licenses for existing assets}
    \item[] Question: Are the creators or original owners of assets (e.g., code, data, models), used in the paper, properly credited and are the license and terms of use explicitly mentioned and properly respected?
    \item[] Answer: \answerYes{} % Replace by \answerYes{}, \answerNo{}, or \answerNA{}.
    \item[] Justification: All creators and owners of assets are properly credited. The license and terms of use are respected and listed in Section~\ref{sec:datasets} and Section~\ref{sec:extensive_model_description} in the Appendix.

\item {\bf New assets}
    \item[] Question: Are new assets introduced in the paper well documented and is the documentation provided alongside the assets?
    \item[] Answer: \answerYes{} % Replace by \answerYes{}, \answerNo{}, or \answerNA{}.
    \item[] Justification: The paper introduces \textsc{QiskitHumanEvalSimX}, detailing the specific prompt template used to generate it via Gemini 3 Pro. The dataset is provided in the code as part of the supplementary materials and will be released under a permissive license.

\item {\bf Crowdsourcing and research with human subjects}
    \item[] Question: For crowdsourcing experiments and research with human subjects, does the paper include the full text of instructions given to participants and screenshots, if applicable, as well as details about compensation (if any)? 
    \item[] Answer: \answerNA{} % Replace by \answerYes{}, \answerNo{}, or \answerNA{}.
    \item[] Justification: The research does not involve crowdsourcing or human subjects.

\item {\bf Institutional review board (IRB) approvals or equivalent for research with human subjects}
    \item[] Question: Does the paper describe potential risks incurred by study participants, whether such risks were disclosed to the subjects, and whether Institutional Review Board (IRB) approvals (or an equivalent approval/review based on the requirements of your country or institution) were obtained?
    \item[] Answer: \answerNA{} % Replace by \answerYes{}, \answerNo{}, or \answerNA{}.
    \item[] Justification: The research does not involve human subjects.

\item {\bf Declaration of LLM usage}
    \item[] Question: Does the paper describe the usage of LLMs if it is an important, original, or non-standard component of the core methods in this research? Note that if the LLM is used only for writing, editing, or formatting purposes and does \emph{not} impact the core methodology, scientific rigor, or originality of the research, declaration is not required.
    %this research? 
    \item[] Answer: \answerYes % Replace by \answerYes{}, \answerNo{}, or \answerNA{}.
    \item[] Justification: The use of LLMs is the core subject of the paper, and the specific models used (Qwen3, GPT-OSS, MiniMax, Gemini 3 Pro) are thoroughly documented.

\end{enumerate}